\documentclass{article}
\usepackage{graphicx} 
\usepackage[a4paper, total={6in, 8in}]{geometry}
\usepackage{amsmath, amsfonts, amsthm}
\usepackage{booktabs}
\usepackage{color}
\usepackage{xcolor}
\usepackage{subcaption}
\usepackage{float}
\usepackage{enumitem} 

\allowdisplaybreaks 
\def\R{\mathbb{R}}

\def\EBE{\mathrm{EBE}}

\def\B{\mathcal{B}}

\def\rE{\mathrm{E}}
\def\rB{\mathrm{B}}
\def\rS{\mathrm{S}}
\def\rT{\mathrm{T}}
\def\rR{\mathrm{R}}
\def\cV{\mathcal{V}} 

\newcommand{\diff}{\text{d}}

\newcommand*\der{\mathop{}\!\mathrm{d}}

\newtheorem{theorem}{Theorem}[section]
\newtheorem{lemma}[theorem]{Lemma}

\newtheorem{corollary}{Corollary}[section]

\newcommand{\s}{\mathfrak{s}}
\newcommand{\q}{\mathfrak{q}}
\numberwithin{equation}{section}
\newtheorem{assumption}{Assumption}

\providecommand{\keywords}[1]
{
  \small	
  \noindent\textbf{\textit{Key words. }} #1
}

\providecommand{\msc}[1]
{
  \small	
  \noindent\textbf{\textit{AMS subject classifications.}} #1
}
\usepackage{mathrsfs}
\title{Switching Hamiltonian Monte Carlo \\ for sampling from mixture distributions}
\author{A. Sharma}
\date{}

\begin{document}

\maketitle
\begin{abstract}
We introduce a  switching Hamiltonian Monte Carlo method for sampling from finite mixture Boltzmann-Gibbs distributions. We propose symmetric numerical integrators to approximate switching Hamiltonian dynamics interlaced with Poisson jumps, where the regime-switching chain is simulated using the uniformization technique or the stochastic simulation algorithm. We prove geometric ergodicity of the resulting Markov chain. We develop an approach based on the discrete Poisson equation associated with numerical schemes to estimate the error in computing ergodic averages. Using this approach we prove that the proposed numerical integrators have second-order bias. This approach is simple and can be generalized to other settings, for example, kinetic Langevin equations. Finally, we verify the convergence result via  numerical experiment.

\end{abstract}


\msc{
65C30, 60H35, 60H10, 37H10}
\section{Introduction}

Sampling techniques are widely used to evaluate risk and measure unknowns, finding parameters within statistical frameworks, compute thermodynamic variables in physics simulations and designing new global optimization methods. On the other hand, regime-switching jump-diffusion stochastic differential equations (SDEs)  find a large number of applications in biology, insurance, control etc. (for more in-depth exposition see \cite{mao2006stochastic_switching,yin2009hybrid}). The underlying process in such hybrid systems has two components $(Y(t), \s(t))$ where $Y$ represents the jump-diffusion and $\s$ describes switching mechanism.  This paper utilizes the framework of switching stochastic dynamics for the purpose of sampling from mixture distributions as in \cite{tretyakov2025sampling}.  Mixture distributions play an important role in statistical modeling and Bayesian inference \cite{fruhwirth2006finite,bouguila2020mixture}.

 Arguably, discrete-time Hamiltonian Monte Carlo \cite{duane1987hybrid} and its variants \cite{bou_serna2018geometric,jansson2025rebalancing} are among the most widely used sampling methods.  In this paper, we take a continuous-time perspective on Hamiltonian Monte Carlo. We consider hybrid-Hamiltonian dynamics which includes a switching mechanism and is interlaced with randomized refreshments. We  write the resulting regime-switching jump dynamics as SDEs driven by Poisson random measure with added regime-switching mechanism in the drift term so that the target measure is a mixture Boltzmann-Gibbs distribution.  Then the contributions of this paper are as follows:
\begin{itemize}
\item[(i)] We propose new numerical integrators based on splitting technique to discretize switching Hamiltonian-type SDEs where continuous-time finite state space conditionally Markov chain is simulated via uniformization technique or Gillespie's algorithm. The proposed schemes improve the error incurred in estimating ergodic averages to $\mathcal{O}(h^2)$ from $\mathcal{O}(h)$ bias resulting from switching Langevin algorithm proposed in \cite{tretyakov2025sampling} (here $h$ represents the discretization step).

\item[(ii)] We present a new approach based on discrete Poisson equation associated with numerical scheme itself for obtaining error bounds in  integral probability metric between desired measure $\mu$ and the invariant measure $\mu_h$ of Markov chain governed by the rules of proposed numerical schemes, i.e.
\begin{align}
    \left|\int_{\mathbb{R}^d} \varphi \der \mu  - \int_{\mathbb{R}^d} \varphi \der \mu_h \right| \leq Ch^2, 
\end{align}
where $\varphi$ belongs to a class of test functions allowing for bounded measurable functions, Lipschitz functions etc., and $C$ is independent of $h$ and is dependent on the solution of a discrete Poisson equation. 
  
\end{itemize}

The other contributions which are of independent interest and serve as the building blocks of the above main results are as follows.  We first show the geometric ergodicity of continuous-time dynamics.   Writing the dynamics as SDEs allows us to use tools available in  stochastic numerics literature to develop numerical integrators for randomized Hamiltonian Monte Carlo dynamics. A lot of attention has been given in designing and analyzing numerical integrators for kinetic Langevin dynamics, for example \cite{milstein2003quasi,burrage2009accurate, stoltz2010free,bourabeeowhadi2010, abdulle2014ergo, leimkuhler2013rational, alamo2016technique, sanz2021wasserstein,monmarche2021high, leimkuhler2024contraction} but the same can not be said about Hamiltonian dynamics interlaced with Poisson refreshments for which we propose splitting schemes and also show their uniform in $h$ geometric ergodicity.  We now position our work within the literature on numerical integration of switching SDEs. Existing convergence results for numerical schemes for regime-switching processes cover several different settings. For switching diffusions without jumps, \cite{yin2010approximation} proved weak convergence of a numerical scheme without rate, \cite{shao2018invariant} established an optimal 1/2 rate in $L^1$ and \cite{song2006numerical_auto} investigated  continuous time Markov chain approximation scheme where the generator itself has to be discretized. For regime-switching jump diffusions, convergence of the Euler method is investigated in \cite{yin2005numerical} in $L^2$
 with decreasing step sizes. Here, we take splitting approach for numerical approximation which allows to use uniformization technique or Gillespie' methods to simulate continuous time chain governing switching mechanism.   

We denote with $(a\cdot b)$ scalar product between two vectors $a, b \in \mathbb{R}^d$ and with $|a|$ norm of $a \in \mathbb{R}^d$. We denote with $\mathcal{N}(0, I_d)$ $d$-dimensional standard Gaussian distribution and with $\mathcal{U}(a,b)$ the uniform distribution on $(a,b)$ with $a< b \in \mathbb{R}$. We represent $I_d$ and $\mathbf{0}_d$ the identity matrix and zero matrix, respectively. For a matrix $P \in \mathbb{R}^{n \times n}$, we denote its matrix norm as $\|P\|$.

\section{Switching Hamiltonian Monte Carlo}

Andersen’s thermostat \cite{andersen1980molecular} models a molecular system using Hamiltonian dynamics with random Poisson jump updates in the momentum variable of a randomly chosen particle. These random refreshments of randomly selected particles imitate collisions. In the same spirit, the authors of \cite{bou2017randomized} introduced randomized Hamiltonian Monte Carlo (RHMC), where Hamiltonian dynamics is interlaced with Poisson refreshments, to sample from the Boltzmann-Gibbs measure defined as
\begin{align}
\mu(\diff x, \diff v) \propto e^{-U(x) - |v|^2/2}  \der x \der v .   
\end{align} In RHMC, starting at time $t_0$, a random time $\delta t$ is drawn from an exponential distribution with mean $1/\lambda$, and the system evolves according to standard Hamiltonian dynamics until $t_1 = t_0 + \delta t$. At time $t_1$, the position remains unchanged, but the velocity is partially refreshed by replacing it with
$$
V(t_1) = \cos(\phi)V(t_1^-) + \sin(\phi)Z,
$$
where $\phi \in (0, \pi/2]$ determines the amount of randomness added and $Z \sim \mathcal{N}(0, I_d)$. We can write the RHMC dynamics in terms of SDEs governed by Poisson random measure 
as follows:
\begin{equation} \label{eqn_Poishmc}
\begin{aligned}
    \der X(t) & = V(t) \der t, \\
    \der V(t) & = -\nabla U(X(t)) \der t + \int_{\mathbb{R}^d}(- V(t^{-}) + \cos(\phi)V(t^{-}) + \sin(\phi) z) N(\der t ,\der z),
\end{aligned}
\end{equation}
where $N(\der t, \der z) $ denotes the Poisson random measure with intensity  $\lambda \nu( \der z)\der t$, $\nu( \der z) \propto e^{-|z|^2/2} \der z$. Under the assumptions of local one-sided Lipschitz continuity and super-linear growth of $-\nabla U$, there exists a unique strong solution of \eqref{eqn_Poishmc} (see \cite{gyongy1980stochastic}). As can be verified $\rho$ satisfies the corresponding Fokker-Planck partial integro-differential equation (PIDE):

\begin{align}
   &- (v \cdot \nabla_x \rho) + (\nabla U(x) \cdot \nabla_v \rho)  
          +    \frac{\lambda}{(2 \pi)^{d/2} \sin^d (\phi)} \int_{\mathbb{R}^d} \rho(x, z) e^{-\frac{|v - \cos(\phi) z|^2}{2 \sin^2(\phi)}}\der z - \lambda\rho(x,v)= 0.
\end{align}

Applications requiring sampling from mixture distributions and the computation of their ergodic averages arise in several statistical models and hybrid-control systems. For this purpose, we present below a switching Hamiltonian Monte Carlo method. Let $\mathcal{M}=\{1,\ldots ,m\}$ denote the finite set of regimes, let $U(x; m) > 0$, $x \in \mathbb{R}^d$, $m \in \mathcal{M}$, be sufficiently smooth functions and consider a given finite mixture distribution with sufficiently smooth density
\begin{equation}
\rho(x)  \propto \sum_{i=1}^{m} \rho(x; i),\;\;\text{
where} \;\;\rho(x; i) = w_i \exp(-U(x; i)),
\end{equation}
$w_i > 0$ are weights. With a slight abuse of notation which should not lead to any confusion, we introduce the mixture Boltzmann-Gibbs density as
\begin{align}
\rho(x,v)  = \frac{1}{\mathcal{Z}} \sum_{i=1}^{m} \rho(x, v; i),\;\; \text{where }\;\; \rho(x,v;i) = w_i \exp(-U(x; i) - |v|^2/2),
\end{align}
and $\mathcal{Z}$ is the normalization constant. 
Let us now introduce the transition-rate matrix governing the regime-switching mechanism.
\begin{equation}\label{swi_hmc_Qmatrix}
{Q}(x)=\left[
\begin{array}{cccc}
-\q_{1}(x) & \q_{12}(x) & \cdots & \q_{1m}(x) \\
\q_{21}(x) & - \q_{2}(x) & \cdots & \q_{2m}(x) \\
\vdots & \vdots & \cdots & \vdots \\
\q_{m1}(x) & \q_{m2}(x) & \cdots & - \q_{m}(x)%
\end{array}%
\right] ,\quad x\in \R^d,
\end{equation}%
where $\q_{i}(x)$ and $\q_{ij}(x)$ are continuous non-negative functions bounded in $\mathbb{R}^{d}$ and the following relation holds:
\begin{equation}
\sum_{j\neq i}\q_{ij}(x)=\q_{i}(x),\ \ i\in \mathcal{M}.  \label{hmc_switching_eqn_q}
\end{equation}
Therefore, the regime-switching jump SDEs governing switching Hamiltonian Monte Carlo dynamics are  
\begin{equation} \label{eqn_switching_Poishmc}
\begin{aligned}
    \der X(t) & = V(t) \der t \\
    \der V(t) & = -\nabla U(X(t); \mathfrak{s}(t)) \der t + \int_{\mathbb{R}^d}(- V(t^{-}) + \cos(\phi)V(t^{-}) + \sin(\phi) z) N(\der t ,\der z),
\end{aligned}
\end{equation}
where $N(\der t, \der z) $ is as in \eqref{eqn_Poishmc} and $\s (t)$ is a
continuous-time finite state conditionally Markov process with right-continuous sample
paths taking values in $\mathcal{M}$ and
\begin{align}
\mathbb{P}(\s (t+\delta ) &= j\mid\s (t)= i ,X(t)=x)=\q_{ij}(x)\delta +o(\delta ),\ i\neq 
 j,  \label{eq:sde2} \\
\mathbb{P}(\s (t+\delta ) &= i\mid \s (t)=i,X(t)=x)=1-\q_{i}(x)\delta +o(\delta )  \notag
\end{align}%
as $\delta \downarrow 0.$ The Poisson random measure $N(\der t, \der z)$ and the
continuous-time conditional Markov chain governing the transitions of $\s (t)$ are independent. Here, the corresponding Fokker-Planck equation with adjoint operator $\mathcal{A}^*$ is
\begin{equation*}
\begin{aligned}
\mathcal{A}^* \rho(x,v; \s) = &  -(v \cdot\nabla_x \rho(x,v;\s)) + (\nabla U(x; \s)\cdot\nabla_v \rho(x,v; \s) )  \\   &  +      \frac{\lambda}{(2 \pi)^{d/2} \sin^d (\phi)} \int_{\mathbb{R}^d} \rho(x, z;\s) e^{-\frac{|v - \cos(\phi) z|^2}{2 \sin^2(\phi)}}\der z - \lambda\rho(x,v;\s) \\
& - \q_\s(x) \rho(x,v; \s) + \sum_{j \neq \s} \q_{j\s}(x) \rho(x, v; j). 
\end{aligned}
\end{equation*}

To ensure that \eqref{hmc_switching_eqn_q} leaves the desired mixture distribution invariant, i.e. $\mathcal{A}^*\rho = 0$ we require the following relation to hold:
\begin{equation*}
-\q_\s(x) \rho(x, v; \s) + \sum_{j \neq \s} \q_{j\s}(x) \rho(x ,v; j) = 0, \quad \s \in \mathcal{M},
\end{equation*}
or equivalently $\sum_{j \neq \s} [\q_{j\s}(x) \rho(x, v; j) - \q_{\s j}(x) \rho(x, v; \s)] = 0$,  $\s \in \mathcal{M}$. One way it is satisfied if for $j \neq \s$
\begin{equation}
\frac{\q_{j\s}(x)}{\q_{\s j}(x)} = \frac{\rho(x , v; \s)}{\rho(x,v; j)}. \label{swi_eqn_2.9}
\end{equation}

 \begin{assumption}\label{assump_1_well_posed}
     $U \in C^{2}(\mathbb{R}^d)$ for each $k \in \mathcal{M}$. There exist constants $L, L_q >0$ such that
     \begin{align}
         |\nabla U(x; k) - \nabla U(y; k)| &\leq L|x-y|, \\
      |\q_{kl}(x)- \q_{kl}(y)| &\leq L_{q}|x-y|,
  \end{align}
  for all $x, y \in \mathbb{R}^d$ and $k\neq l \in \mathcal{M}$. Moreover, $\q_{kl}$ are bounded in $\mathbb{R}^d$.
 \end{assumption}
\begin{assumption}\label{assump_2_ergo}
    There exist constants $\kappa >0$ and $\ell \in \R$ such that the following holds:
    \begin{align}
        -(x \cdot \nabla U(x; i)) \leq - \kappa |x|^2 + \ell,
    \end{align}
    for all $x \in \mathbb{R}^d$ and $i \in \mathcal{M}$. Moreover, we assume $\q_{ij}(x) >0$ for all $i,j \in \mathcal{M}$ with $i \neq j$ and $x \in \mathbb{R}^d$. 
\end{assumption} 

\paragraph{Geometric ergodicity.} Under Assumption~\ref{assump_1_well_posed}, the regime-switching dynamics with jumps \eqref{eqn_switching_Poishmc} is well-posed, i.e. there exists a unique strong solution \cite{fubaoxi2009asymptotic, fubaoxi2011jump}. To establish geometric ergodicity via Meyn-Tweedie approach, we need two ingredients: the Lyapunov drift condition and the minorization condition. The Lyapunov drift condition follows from \cite{bou2017randomized} under Assumption~\ref{assump_2_ergo}. However, to prove minorization condition, diffusion matrix is explicitly used in \cite[Sections~5-6]{fubaoxi2009asymptotic} (also see \cite{fubaoxi2011jump, xi2017feller}). Likewise ellipticity or local ellipticity of diffusion matrix is used in \cite{chen2019properties}. Therefore, we can not rely on these quoted results for geometric ergodicity of \eqref{eqn_switching_Poishmc}. However, \eqref{eqn_switching_Poishmc}, for a fixed switch $\s(t) =k$ for all $t$, satisfies minorization condition \cite[Proposition~3.7]{bou2017randomized}. We will use this fact to show minorization condition for \eqref{eqn_switching_Poishmc} and thus establishing its geometric ergodicity. 

Let us consider \eqref{eqn_switching_Poishmc} again. We fix the switches $\s(t) = k$ and $\s(t)  = l$ for all $t \geq 0$ and denote the fixed switch (i.e., without any regime-switching mechanism) processes with $(X^{(k)} , V^{(k)})$ and $(X^{(l)}, V^{(l)})$, respectively. Thanks to \cite[Proposition~3.7]{bou2017randomized}, there exist times $t_1  > 0$ and $t_3  > 0$ such that the unswitched processes $(X^{(k)} , V^{(k)})$ and $(X^{(l)}, V^{(l)})$ satisfy minorization on a compact set $\mathcal{C}$. We take $t_2 >0$ to be an arbitrary intermediate transition time whose choice we will make later. Let $t^* = t_1 + t_2 + t_3$. We apply the Chapman-Kolmogorov equation twice, expanding over all possible intermediate states $(y, m)$ at time $t_1$, and $(z, n)$ at time $t_1 + t_2$, to get 
\begin{align*}
    P_{t^*}&\big((y_0, k), A \times \{l\}\big) \\
    &= \sum_{i \in \mathcal{M}} \sum_{n \in \mathcal{M}} \int_{\mathbb{R}^{2d}} \int_{\mathbb{R}^{2d}} P_{t_1}\big((y_0, k), \mathrm{d}y \times \{i\}\big) P_{t_2}\big((y, i), \mathrm{d}z \times \{n\}\big) P_{t_3}\big((z, n), A \times \{l\}\big).
\end{align*}
 Restricting the first sum to $i = k$, the second sum to $n = l$, and the spatial integral over $z$ to the compact set $\mathcal{C}$ yields
 \begin{align}\label{eq:CK_bound}
    P_{t^*}&\big((y_0, k), A \times \{l\}\big) \nonumber\\
    & \geq \int_{\mathbb{R}^{2d}} \int_{\mathcal{C}} P_{t_1}\big((y_0, k), \mathrm{d}y \times \{k\}\big) P_{t_2}\big((y, k), \mathrm{d}z \times \{l\}\big) P_{t_3}\big((z, l), A \times \{l\}\big).
\end{align}
We  can get a lower bound on first and third terms of the integrand by conditioning on the event of no switching.  
For the third term, starting at $z \in \mathcal{C}$ and remaining in $l$ for duration $t_3$, we apply the unswitched minorization condition to arrive at
\begin{equation}\label{eq:bound3}
    P_{t_3}\big((z, l), A \times \{l\}\big) \geq e^{-\bar{q}_{\max} t_3} P^{(l)}_{t_3}(z, A) \geq e^{-\bar{q}_{\max} t_3} \epsilon_{l} \nu_{l}(A),
\end{equation}
where $\bar{q}_{\max} = \sup_{x \in \mathbb R^d}\max_{i \in \mathcal{M}}\q_i(x)$.  
Similarly, for the first term, conditioning on starting at $y_0 \in \mathcal{C}$ and remaining in $k$ for duration $t_1$ yields 
\begin{equation}\label{eq:bound1}
    P_{t_1}\big((y_0, k), \mathrm{d}y \times \{k\}\big) \ge e^{-\bar{q}_{\max} t_1} P^{(k)}_{t_1}(y_0, \mathrm{d}y) \geq e^{-\bar{q}_{\max} t_1} \epsilon_{k} \nu_{k}(\mathrm{d}y).
\end{equation}
Using \eqref{eq:bound3} and \eqref{eq:bound1} in \eqref{eq:CK_bound} gives
\begin{align*}
    P_{t^*}&\big((y_0, k), A \times \{l\}\big) 
    \geq e^{-\bar{q}_{\max}(t_1+t_3)} \epsilon_k \epsilon_{l} \nu_{l}(A) \int_{\mathbb{R}^{2d}} P_{t_2}\big((y, k), \mathcal{C} \times \{l\}\big) \nu_k(\mathrm{d}y).
\end{align*}
 Let $B_k\subset\operatorname{int}\mathcal{C}$, i.e. interior $\mathcal{C}$. From \cite[Proposition~3.7]{bou2017randomized}, we know that $\nu_k$ is a probability measure on $\mathcal{C}$ with $\nu_k(B_k)>0$. By
continuity of the Hamiltonian flow, for $t_2>0$ sufficiently small, the continuous component starting from any $y\in B_k$ remains in $\mathcal{C}$ on
$[0,t_2]$ conditional on the event that there is no velocity-refreshment jump, which has probability
$e^{-\lambda t_2}>0$. If $l=k$, the event of no switching of regime and no velocity-refreshment jump gives that $P_{t_2}\big((y,k),\mathcal C\times\{k\}\big)>0$, $y\in B_k$. If $l\neq k$, then using continuity and strict positivity of $\q_{kl}$, the
event that exactly one switch $k\to l$ occurs during $[0,t_2]$, no further
switch occurs, and no velocity-refreshment jump occurs has strictly positive
probability. On this event the continuous component remains in $\mathcal C$
and the terminal regime is $l$. Hence, $P_{t_2}\big((y,k),\mathcal{C}\times\{l\}\big)>0$, $y\in B_k$. Consequently,
\begin{align}
\int_{\mathbb R^{2d}}
P_{t_2}\big((y,k),\mathcal C\times\{l\}\big)\nu_k(\mathrm dy)
\geq \int_{B_k}
P_{t_2}\big((y,k),\mathcal{C}\times\{l\}\big)\nu_k(\der y)>0.
\end{align}
Therefore, 
\begin{align*}
     P_{t^*}&\big((y_0, k), A \times \{l\}\big)  \geq \epsilon \tilde{\nu}(A\times\{l\}),
\end{align*}
where $\tilde{\nu}(\der z, \{j\}) = \nu_l(\der z) I_{j = l} $ and
$ 
\epsilon = \min_{k,l \in \mathcal{M}}e^{-\bar{q}_{\max}(t_1+t_3)} \epsilon_k \epsilon_{l} \int_{B_k} P_{t_2}\big((y, k), \mathcal{C} \times \{l\}\big) \nu_k(\mathrm{d}y).
$

\section{Numerical integrators}
In Subsections~\ref{eb_subsection} and~\ref{ebs_subsection}, we present numerical integrators for randomized HMC dynamics \eqref{eqn_Poishmc} and regime-switching HMC dynamics \eqref{eqn_switching_Poishmc}, respectively.  
\subsection{[E, B] schemes }\label{eb_subsection}
Splitting techniques are a standard way to construct higher-order weak or strong approximations of stochastic dynamics. The idea is to decompose the  SDEs into components that are either exactly solvable or computationally straightforward to approximate. This approach has successfully been utilized for Hamiltonian ODEs and Langevin SDEs \cite{leimkuhler2013rational, bourabeeowhadi2010,alamo2016technique}. In this subsection, we consider the SDEs \eqref{eqn_Poishmc} and partition the dynamics in the following manner:
\begin{equation}
\begin{pmatrix} \der X(t) \\ \der V(t) \end{pmatrix} = \underbrace{\begin{pmatrix} 0 \\ -\nabla U(X(t))\der t \end{pmatrix}}_{\mathcal{B}} + \underbrace{\begin{pmatrix} V(t) \der t \\ \displaystyle\int_{\mathbb{R}^d} (-V(t^-) + \cos(\phi) V(t^{-}) + \sin( \phi) z) N(\der t, \der z) \end{pmatrix}}_{\mathcal{E}},
\end{equation}
where each of these parts can be solved exactly. 

Fix a final time $T$. Consider a uniform partition of $[0,T]$ with discretization parameter $h \in (0,1)$ such that the number of steps is $\lfloor T/h\rfloor$. Let $N_t$ denote the Poisson process with intensity $\lambda$.  Let $N_h $ be the total number of jumps in the interval $[t_k, t_k + h)$ and conditioned on $N_h = n$, let $\tau_1, \tau_2, \dots, \tau_{n}$ be the exact times when jumps occur which are distributed as the order statistics simulated from $\mathcal{U}(t_k, t_k+h)$. We define the boundaries as $\tau_0 = t_k$ and $\tau_{n+1} = t_k +h$. Let $Z_i \sim \mathcal{N}(0, I_d)$, $i = 1,\dots, n$. On any given interval $ [\tau_i, \tau_{i+1})$, $i =0,\dots, n$, the velocity in $\mathcal{E}$ is  constant and given by
\begin{equation}
    \zeta_i = v \cos^{i}(\phi) + \sin(\phi)\sum_{j =1}^{i}\cos^{i-j}(\phi)Z_j.
\end{equation}
This also means that integral $\int_{0}^{h}V(s) \der s$ is simply a summation over the inter-jump intervals. Having all the ingredients, we can write the exact solution of $\mathcal{B}$ and $\mathcal{E}$ as follows: 
\begin{align}
    \mathscr{S}_h^{\mathrm{B}}(x,v) &= (x, v - h\nabla U(x)), \label{B_update_swi_HMC} 
     \\
    \mathscr{S}_h^{\mathrm{E}}(x,v) &= \left(x + \textstyle\sum\limits_{i=0}^{N_h}\zeta_{i} (\tau_{i+1} - \tau_{i}), \zeta_{N_h}\right),
\end{align}
and, for brevity, we use same notation to denote the conditional update
\begin{align}
    \mathscr{S}_h^{\mathrm{E}}(x,v; N_h = n) &= \left(x + \textstyle\sum\limits_{i=0}^{n}\zeta_{i} (\tau_{i+1} - \tau_{i}), \zeta_{n}\right).
\end{align}
Since our main interest is in sampling, the update rule for $\mathcal E$ step can be simplified via the Cholesky decomposition. Specifically, the same joint distribution of the update in $(x,v)$ can be preserved by generating two Gaussian random vectors $\xi_1, \xi_2 \sim \mathcal{N}(0, I_{d})$, providing the following representation: 
\begin{align}\label{E_update_swi_HMC}
      \mathscr{S}_h^{\mathrm{E}} &= \left(x +  m_x + \sin(\phi)\bigg( \frac{\sigma_{xv}}{\sqrt{\sigma_{vv}}} \xi_1 + \sqrt{\sigma_{xx} - \frac{\sigma_{xv}^2}{\sigma_{vv}}} \xi_2\bigg) , v\cos^{n}(\phi) + \sqrt{1 - \cos^{2n}(\phi)}\xi_1\right),
\end{align}
where
\begin{align*}
m_x = v\sum_{i=0}^{n}\cos^{i}(\phi)(\tau_{i+1} - \tau_{i}), \quad   \sigma_{vv} = \frac{1 - \cos^{2n}(\phi)}{\sin^2(\phi)}, \quad 
    \sigma_{xx}  = \sum_{l=1}^{n}a_l^2, \quad
    \sigma_{xv}  =\sum_{l=1}^{n} a_{l}\cos^{n-l}(\phi),
\end{align*}
with $ a_l = \sum_{i = l}^{n}(\tau_{i+ 1} - \tau_{i})\cos^{i-l}(\phi) $. Therefore, $a_{l} = \tau_{l+1}- \tau_{l} + \cos(\phi) a_{l+1}$ with $a_n = t_k + h - \tau_{n}$. We assign $\sigma_{vv} = 1$ and $\sigma_{xx} = \sigma_{xv} = 0$ if $n = 0$.
\begin{equation}
\begin{array}{cc}
 \text{[EBE]} & \text{[BEB]} \\ [6pt]
\begin{aligned}
(x_{1/2}, \hat{v}_{1/2}) &= \mathscr{S}_{h/2}^{\mathrm{E}}(x, v), \\
(x_{1/2},v_{1/2}) &= \mathscr{S}_{h}^{\mathrm{B}}(x_{1/2}, \hat{v}_{1/2}), \\
(x_{1}, v_{1}) &= \mathscr{S}_{h/2}^{\mathrm{E}}(x_{1/2}, v_{1/2}).
\end{aligned}
&
\begin{aligned}
(x,\hat{v}_{1/2}) &= \mathscr{S}_{h/2}^{\mathrm{B}}(x, v), \\
(x_{1}, v_{1/2}) &= \mathscr{S}_{h}^{\mathrm{E}}(x, \hat{v}_{1/2}), \\
(x_1,v_{1}) &= \mathscr{S}_{h/2}^{\mathrm{B}}(x_{1}, v_{1/2}).
\end{aligned}
\end{array}
\end{equation}

Instead of decomposing into deterministic flow and jump component as is done for PDMPs \cite{bertazzi2025piecewise}, in the proposed schemes above we split the deterministic drift into two subflows and include one of them in the jump component leading to $\mathrm{[BEB]}$ and $\mathrm{[EBE]}$ integrators. This construction is motivated by viewing RHMC  as a stochastic differential equation with jumps \eqref{eqn_Poishmc}. In contrast, the papers devoted to numerical integrators underlying Hamiltonian Monte Carlo or its variants focus mainly on the discretization of Hamiltonian dynamics over a fixed duration as the refreshments arrive at fixed times (see for example \cite{blanes2014numerical, akhmatskaya2009comparison, akhmatskaya2008gshmc, bou_serna2018geometric}).  The presented integrators have clear resemblance to $\mathrm{[UBU]}$ integrator for discretizing kinetic Langevin SDEs proposed in \cite{alamo2016technique} and investigated in \cite{sanz2021wasserstein,chada2023unbiased}. This means that several new numerical integrators can be designed for \eqref{eqn_Poishmc} as has been done for kinetic Langevin dynamics. Such a study  will be of interest considering path-wise properties of \eqref{eqn_Poishmc} are different from kinetic Langevin SDEs.

\subsection{[E, B, S] schemes}\label{ebs_subsection}
Here, we include the switching mechanism in the [E, B] schemes introduced in the previous subsection. During the S-step, the position $x$ is kept fixed. This results in a continuous-time  Markov chain (CTMC) $\s(t)$ with finite state space whose generator matrix is $Q$ (from \eqref{swi_hmc_Qmatrix}).
Let $\Pi(t)$ denote the row vector of probabilities of the chain at time $t$. Then, its evolution is described by
\begin{align}
\frac{\der }{\der t}\Pi(t) = \Pi(t) Q(x),
\qquad 
\Pi(t) = \Pi(0) \exp({t Q(x)}), \quad t \in [0,h], \label{swi_hmc_3.50}
\end{align}
where $\exp(\cdot)$ represents the matrix exponential. Consequently, the semigroup generated by $\mathrm{S}$ acts on functions 
$g = g(x,v,m)$ as
\begin{equation}
(P_t^\mathrm{S} g)(x,v,m)
=
\sum_{j \in \mathcal{M}} (\exp({tQ(x)}))_{mj} \, g(x,v,j).
\end{equation}
 To incorporate the switching mechanism in the [E, B] scheme, we need a random sample from $\Pi(h)$ i.e. a sample from the probability distribution of the continuous time Markov chain $\s$ (note that $x$ is fixed) at time $h$.  The practicability of the above solution is limited by the required computational cost and memory usage if the size of matrix $Q$ is large. There have been several strategies proposed in the literature to simulate CTMC : (i) Gillespie method \cite{gillespie1976general,ross2022simulation,reibman1989markov}, commonly known as stochastic simulation algorithm (SSA) or direct jump method, and (ii) Uniformization method \cite{jensen1953markoff}. We can also use a higher order approximation which induces $\mathcal{O}(h^\alpha)$, $\alpha \geq 3$ error in each step and hence maintaining second order of convergence of resulting [E, B, S] schemes. Since we are interested in a random sample at fixed time $h$, we will first discuss the uniformization technique to simulate S step. This is exact in distribution like the Gillespie method. The idea behind the algorithm is to  transform the CTMC, where transition to next state depends on current state, to a discrete time Markov chain via random time change. Here, for the sake of completeness, we provide a brief derivation. Let $P = I +  \frac{1}{\eta}Q $ where $\eta \geq \max_{i} \q_i$. This implies 
\begin{align*}
    \Pi(t) &= \Pi(0) \exp((\eta P - \eta I )t) = \Pi(0) \exp(-\eta I t) \exp(\eta P t) \\  & = e^{-\eta t}\Pi(0) \exp(\eta P t) 
 = e^{-\eta t}\sum_{n=0}^{\infty} \frac{(\eta t)^{n}}{n !}\Pi(0)P^n. 
\end{align*}
The uniformization allows us to simulate a continuous time Markov chain with generator matrix $Q$ by simulating discrete time Markov chain with transition matrix $P$. If $\s$ is our original continuous time Markov chain and we denote the discrete time Markov chain with transition matrix  $P$ as $\mathfrak{m}$, then $\s(t) = \mathfrak{m}_{N^{\eta}(t)}$ where $N^{\eta}(t)$ is a Poisson clock with intensity $\eta$. Let $P_{ij}$ denote the $ij$-th entry in matrix $P$. Therefore, given $x$, the update in switching Markov chain is given by 
\begin{equation}
\mathscr{S}_h^{\mathrm{S}}(\mathbf{s}) = \mathfrak{m}_{n}, \text{ where }
\begin{cases} 
\text{simulate } n := N^{\eta} \sim \text{Poisson}(\eta h), \\ 
\mathfrak{m}_i \sim P_{\mathfrak{m}_{i-1} \cdot} \text{ for } i = 1, \dots, n, & P = I + \frac{1}{\eta}Q,\; \mathfrak{m}_0 = \mathbf{s},
\end{cases}
\end{equation}
where $\mathfrak{m}_i \sim P_{\mathfrak{m}_{i-1} \cdot}$ means
\begin{equation}
\mathfrak{m}_i = \min \left\{ j \in \mathcal{M} \ \Bigg| \ \sum_{l=1}^{j} P_{\mathfrak{m}_{i-1} l} \geq u_i \right\}, \text{ with } u_i \sim \mathcal{U}(0, 1).
\end{equation}
Having introduced the uniformization technique above, we may now combine it with the $[\text{E,\;B}]$ schemes in order to define the class of $[\text{E,\;B,\;S}]$ schemes for switching Hamiltonian Monte Carlo dynamics \eqref{eqn_switching_Poishmc}. In this setting, there can be six possible symmetric concatenations, namely
[SEBES], [SBEBS], [BESEB], [BSESB], [ESBSE] and [EBSBE].
Except [SBEBS] and [EBSBE] all other schemes require $M+1$ gradient evaluations for $M$ steps. We write below two possible concatenations [SEBES] and [BESEB]. 
\begin{equation}\label{swi_hmc_sebes}
\begin{array}{cc}
 \text{[SEBES]} & \text{[BESEB]} \\ [6pt]
\begin{aligned}
\s_1 &=  \mathscr{S}_{h/2}^{\mathrm{S}}(\s_0\; ; \; x), \\
(x_{1/2}, \hat{v}_{1/2}) &= \mathscr{S}_{h/2}^{\mathrm{E}}(x, v), \\
(x_{1/2}, v_{1/2}) &= \mathscr{S}_{h}^{\mathrm{B}}(x_{1/2}, \hat{v}_{1/2}\; ;\; \s_1), \\
(x_{1}, v_{1}) &= \mathscr{S}_{h/2}^{\mathrm{E}}(x_{1/2}, v_{1/2}),\\
\s_2 &=  \mathscr{S}_{h/2}^{\mathrm{S}}(\s_1\; ; \; x_1).
\end{aligned}
&
\begin{aligned}
(x, \hat{v}_{1/2}) &= \mathscr{S}_{h/2}^{\mathrm{B}}(x, v\;;\; \s_0), \\
(x_{1/2}, v_{1/2}) &= \mathscr{S}_{h/2}^{\mathrm{E}}(x, \hat{v}_{1/2}), \\
\s_1 &=  \mathscr{S}_{h}^{\mathrm{S}}(\s_0\; ; \; x_{1/2}), \\
(x_1, \hat{v}_{1}) &= \mathscr{S}_{h/2}^{\mathrm{E}}(x_{1/2}, v_{1/2}),\\
 (x_1, v_1) &= \mathscr{S}_{h/2}^{\mathrm{B}}(x_1, \hat{v_1}\; ; \; \s_1).
\end{aligned}
\end{array}
\end{equation}

If $\eta$ is large, simulation based on uniformization may have large number of virtual jumps. 
Instead of uniformizing the chain, we may  simulate its jump times via Gillespie's construction.  Recall that $\q_\s = \sum_{j\neq \s }\q_{\s j}$ denotes the total outgoing rate from state $\s$.  Starting from $\mathfrak{m}_0 = \s $, $\hat{t}_0 = t_k$, define recursively
\begin{align}\label{swi_hmc_3.13}
    \tau_{1} = \mathcal{E}(\q_{\mathfrak{m}_{0}}),\quad \hat{t}_{1} := \hat{t}_0 + \tau_{1}.
\end{align}
If $\hat{t}_1 > t_k + h$ then the chain does not jump during the interval $[t_k,t_k +h]$, and the S step is simply $\mathscr{S}_h = \mathfrak{m}_{0} = \s$. Else the next regime is sampled as
\begin{align}\label{swi_hmc_3.14}
    \mathbb{P}(\mathfrak{m}_{1} = j\; | \; \mathfrak{m}_{0} = \s, x) = \frac{\q_{\s j}(x)}{\q_\s(x)}\qquad j \neq \s.
\end{align}
Equivalently, with \(u_1\sim\mathcal U(0,1)\),
\begin{align}\label{swi_hmc_3.15}
\mathfrak{m}_{1} =
\min\left\{
j\in\mathcal{M}\setminus\{\s\}
\;\Bigg|\;
\sum_{\ell\in\mathcal M\setminus\{\s\},\,\ell\leq j}
\frac{\q_{\s\ell}(x)}{\q_{\s}(x)}
\ge u_1
\right\}.
\end{align}
The procedure \eqref{swi_hmc_3.13}-\eqref{swi_hmc_3.15} is repeated for $i =1,\dots, n$ where $n$ is the largest integer such that $\hat{t}_{n}\le t_k + h$, i.e.  until the next proposed jump time exceeds $t_k + h$. Therefore, the Gillespie update can be written as
\begin{align}
\mathscr{S}_h^{\mathrm{S}}(\mathbf{s})
=
\mathfrak{m}_{n}.
\end{align}
We can simply replace the S step in the [E, B, S] schemes with Gillespie's construction. 

One simplification is possible if we take the switching rates of the form
\begin{align}
    \q_{ij} = c p_j, \quad i \neq j,
\end{align}
where $c$ is some positive constant, $p_j \geq 0$ and $\sum_{j=1}^m p_j = 1$. As is clear, we have
\begin{align}
    \q_{ii} = -\sum_{j \neq i}\q_{ij} = - c(1 -p_i).
\end{align}
Then the generator can be written as $
Q(x)=c\bigl(\mathbf{1}_m p(x)^{\top}-I_m\bigr)$
where $p(x)=(p_1(x),\ldots,p_m(x))^\top$ and $\mathbf{1}_m \in \mathbb{R}^{m}$ with each component $1$. Since
\(\mathbf{1}p(x)^\top\) is an idempotent matrix, the transition matrix is explicit, i.e.
\begin{align}
\exp(hQ(x))
=
e^{-ch}I_m+\bigl(1-e^{-ch}\bigr)\mathbf{1}_m p(x)^\top .
\end{align}
In component-wise form, we have
\begin{align}
\mathbb{P}(\s_1=j\mid \s_0=i,x)
= e^{-ch}I({\{i=j\}}) + \bigl(1-e^{-ch}\bigr)p_j(x).
\end{align}
Thus the S step can be sampled exactly without 
uniformization or Gillespie method. 
\paragraph{Comparison with switching Langevin algorithm from \cite{tretyakov2025sampling}.}
  To sample from mixture distributions, the following stochastic differential equation with state-dependent switching is considered in \cite{tretyakov2025sampling}:
    \begin{align*}
        &\der X(t) = -\frac{1}{2}\nabla U(X(t); \s(t))\der t + \der B(t), \;\;\mathbb{P}(\s(t+\delta)=j \mid \s(t)=i) = \q_{ij}(X(t))\delta + o(\delta), \quad j \neq i,
    \end{align*}
    where $\s(t)$ is the conditionally Markov switching state and $B(t)$ is a standard Brownian motion, and the switching rates $\q_{ij}(x)$ are from \eqref{swi_hmc_Qmatrix}.   Switching Langevin algorithm proposed in \cite{tretyakov2025sampling} is based on  an explicit Euler scheme with constant time step $h$:
    \begin{align}\label{swi_lang_algo_1}
        X_{k+1} &= X_k - \frac{h}{2} \nabla U(X_k; \s_k) + \sqrt{h} \xi_{k+1}, \quad  \xi_{k+1} \sim \mathcal{N}(0, I_d),
    \end{align}
    where state-dependent regime-switching chain is approximated as
    \begin{align}\label{swi_lang_algo_2}
        \s_{k+1} &= 
        \begin{cases} 
            j \neq \s_k & \text{with prob } \q_{\s_k j}(x_k)h, \\ 
            \s_k & \text{with prob } 1 - \q_{\s_k}(x_k)h. 
        \end{cases}
    \end{align}
In comparison with \eqref{swi_lang_algo_2}, we use uniformization technique or Gillespie's algorithm to simulate regime changes. Moreover, thanks to splitting technique relying on exact simulation of switching chain, error incurred in $\mathrm{[E, B, S]}$ scheme is $\mathcal{O}(h^2)$ (with one force evaluation in most of the concatenations) as compared to $\mathcal{O}(h)$ in \eqref{swi_lang_algo_1}-\eqref{swi_lang_algo_2}.

\section{Main Results}
Our first main result establishes geometric ergodicity of numerical scheme $\mathrm{[SEBES]}$ proposed in previous section for \eqref{eqn_switching_Poishmc}. We ensure that there exists a constant $h_0 \in (0,1)$ such that the geometric ergodicity 
result is uniform in $h \in (0,h_0)$. This is along the lines of \cite{bou2013nonasymptotic,lemakstoltz2016} in contrast to geometric ergodicity  obtained in \cite{mattingly2002ergodicity,bourabeeowhadi2010} for a fixed discretization parameter.

We will use the following Lyapunov function in our analysis \cite{mattingly2002ergodicity,bou2017randomized}:
\begin{align}\label{eqn_swi_hmc_lyapunov}
W (x,v)  = 1 + \mathcal{V}(x,v) = 1 + a|x|^2 + b|v|^2 + c( x \cdot v)  \text{ with } ab > c^2/4.
\end{align}
The further conditions needed on $a, b$ and $c$ are discussed in the proof of Lemma~\ref{lemma_lyapunov_ebe}.

\begin{theorem}[Geometric ergodicity of $\mathrm{[SEBES]}$]\label{sebes_geom_erg_thrm} Let Assumptions~\ref{assump_1_well_posed}-\ref{assump_2_ergo} hold. 
Let the following condition be satisfied:
    \begin{equation}\label{condi_swi_hmc_lambda}
        \lambda > \max\left( \frac{2 \sqrt{2}L/\sqrt{\kappa}}{\sin^2\phi}, \frac{\sqrt{2}L /\sqrt{\kappa}}{2(1 - \cos \phi)}  \right).
    \end{equation}
Then the Markov chain generated by the $\mathrm{[SEBES]}$ scheme is geometrically ergodic. In particular, its one-step transition kernel $\mathcal P^{SEBES}_h$ admits an invariant probability measure, denoted by $\mu_h$, on $\mathbb R^{2d}\times\mathcal M$, and there exist constants $K>0$ and $\gamma > 0$ such that, for every measurable $\varphi$ satisfying
\begin{equation}
|\varphi(x,v,i)|\le W(x,v),
\qquad (x,v,i)\in \mathbb R^{2d}\times\mathcal{M},
\end{equation}
we have
\begin{align}
\left| \mathbb{E}\big[\varphi(X_n,V_n,\s_n)\mid (X_0,V_0,\s_0)=(x,v,i)\big]
-\mu_h(\varphi)
\right|\leq
K W(x,v) e^{-n \gamma h},
\end{align}
where $W$ is a strictly positive Lyapunov function from \eqref{eqn_swi_hmc_lyapunov}.
\end{theorem}
The theorem is proved once we prove the following two lemmas. 
\begin{lemma}
\label{lemma_uniform_minor_sebes}
Let Assumptions~\ref{assump_1_well_posed}-\ref{assump_2_ergo} hold.  Let $\mathcal{P}_{h}^{\mathrm{SEBES}}$ denote the one-step transition kernel of the
\(\mathrm{[SEBES]}\) scheme on $\mathbb R^{2d}\times\mathcal M$. Let \(\mathcal{C}\subset\mathbb R^{2d}\) be compact.  Then there exist $\varepsilon>0$, $h_0 > 0$,
\(\alpha>0\), and a probability measure $\nu$ on $\mathbb R^{2d}\times\mathcal M$
such that, with $
K_\varepsilon:=\left\lceil \frac{\varepsilon}{h}\right\rceil
$,
we have, for all $0 < h \leq h_0$, 
\begin{equation}
(\mathcal{P}_{h}^{\mathrm{SEBES}})^{K_\varepsilon}((x,v,i), O \times \{k\}) 
\geq
\alpha \nu(O \times \{k\}),\quad O\in  \mathscr{B}(\mathbb{R}^{2d}), \; k \in  \mathcal{M}, \; (x,v) \in \mathcal{C}, i\in \mathcal{M}.  
\end{equation}
\end{lemma}

\begin{lemma}\label{lemma_lyapunov_ebe}
   Let the hypotheses of Theorem~\ref{sebes_geom_erg_thrm} hold. Then, there exists a function $W: \mathbb{R}^{2d} \to [1, \infty)$ and constants $c_1 \in (0, 1)$, $c_2 < \infty$ and $h_0 \in (0,1)$ with $c_1 h_0< 1 $ such that $
    \mathbb{E} W(X_{k+1}, V_{k+1}|X_k ,V_k) \leq ( 1- c_1 h) W(X_k, V_k) + c_2 h
$, where $(X_{k+1}, V_{k+1})$ are computed using $\mathrm{[SEBES]}$, 
and thus the following holds:
\begin{equation}
\mathbb{E}W(X_{K_\varepsilon}, V_{K_\varepsilon}) \leq e^{-c_1 K_{\varepsilon} h} W(x, v) + \frac{c_2}{c_1},
\end{equation}
where $K_{\varepsilon} =\left\lceil \frac{\varepsilon}{h}\right\rceil$. 
\end{lemma}

The next result is establishing the order of convergence of the proposed schemes when employed to compute ergodic averages. To this end, we need some notation. For functions on $\mathbb R^{2d}\times\mathcal M$, define
\begin{align*}
L^\infty_\cV
:= \left\{ \varphi : \sup_{(x,v,i)} \frac{|\varphi(x,v,i)|}{1+\cV(x,v)} < \infty \right\},\; L^{1,\cV} := \left\{ g : \sum_{i\in\mathcal M}
\int_{\mathbb R^{2d}} |g(x,v,i)|(1+\cV(x,v))\,\der x\,\der v < \infty \right\}.
\end{align*}
We denote the associated norms by $|\varphi|_{L^\infty_\cV}$ and $ |g|_{L^{1,\cV}}$, respectively.
Let $f\in C_c^\infty(\mathbb R^{2d} \times \mathcal{M}) $. For fixed $(x,v,i)$, the operators $\rE$ and $\rB$,  corresponding to $\mathscr{S}^{\rE}$ and $\mathscr{S}^{\rB}$ steps, act on the $(x,v)$ variables only as
\begin{align*}
(\rB f)(x,v,i)
&=
-(\nabla U(x;i)\cdot \nabla_v f(x,v,i)),\\
(\rE f)(x,v,i)
&=
(v\cdot \nabla_x f(x,v,i))
+\lambda\big(\mathcal R f(x,v,i)-f(x,v,i)\big),
\end{align*}
where the refreshment operator $\mathcal R$ is $(\mathcal R f)(x,v,i)
=
\int_{\mathbb R^d}
f\big(x,\,v\cos\phi+\sin\phi\,z,\,i\big) G(z)\der z$
with $G$ being the standard Gaussian density on $\mathbb R^d$. Likewise, the operator corresponding to switching chain is
\begin{equation}
(\rS f)(x,v,i)
=
-\q_i(x)f(x,v,i)+\sum_{j\neq i} \q_{ij}(x) f(x,v,j).
\end{equation}
Since $\rE$ contains a nonlocal refreshment term, we interpret $[\rE,\rB]$ as the commutator of linear operators rather than as a Lie bracket of vector fields. Namely, on the common test space $C_c^\infty(\mathbb R^{2d})$,
\begin{equation}
[\rE,\rB]f := \rE(\rB f)-\rB(\rE f).
\end{equation}
More generally, all nested brackets appearing below are understood in this operator-theoretic sense.

We can write the corresponding Markov semigroups as
\begin{align*}
(P^{\mathrm{E}}_tf)(x,v) = \sum_{n=0}^\infty e^{-\lambda t} \frac{(\lambda t)^n}{n!} \mathbb{E} \big[ f\big( \mathscr{S}_t^{\mathrm{E},(n)}(x, v) \big) \mid N_t = n \big], \quad 
(P^{\mathrm B}_t f)(x,v) = f\big(x, v-t\nabla U(x)\big),
\end{align*}
and 
\begin{equation}
(P^{\rS}_t f)(x,v,i) = \sum_{j\in\mathcal M} (\exp{(tQ(x))})_{ij} f(x,v,j).
\end{equation}

To quantify the discretization bias in estimating $\mu(\varphi)$, where $\varphi\in L_\cV^\infty$,  we use the following discrete Poisson equation  associated with the geometrically ergodic Markov chain on $\mathbb R^{2d}\times\mathcal{M}$ whose one-step transition operator we denote with $\mathcal{P}_h$:
\begin{align}\label{dis_Poisson_equation_SEBES}
\Big(\frac{\mathcal{P}_h-I}{h}\Big)\Psi
=
\varphi-\mu_h(\varphi).
\end{align}
This equation admits a solution \cite{glynn1996liapounov}, i.e. $\Psi
=
-h\sum_{k=0}^\infty (\mathcal{P}_h)^k\big(\varphi-\mu_h(\varphi)\big)$,
and due to uniform in $h$ geometric ergodicity result (i.e. Theorem~\ref{sebes_geom_erg_thrm}), we also have the following bound:
\begin{align}\label{swi_hmc_poiss_bound}
    |\Psi(x,v,i)|\le \frac{hK}{1- e^{-\gamma h}}\,|\varphi|_{L_\cV^\infty}\,(1+\cV(x,v)),
\end{align}
where $K$ and $\gamma$ are independent of $h$. 
We impose the following assumption.
\begin{assumption}\label{swi_decay_assump}
Let $U(x; i) \in C^{8}(\mathbb{R}^{d})$ for every $i \in \mathcal{M}$. Let $\q_{ij} \in C^{8}(\mathbb{R}^{d})$.     We assume that for every polynomial $\mathrm{Poly}(x,v)$ that can arise from derivatives of $U$, $|v|^2$, and $q_{ij}$ up to order $8$, the following holds:
$$\sum_{i\in\mathcal{M}} \int_{\mathbb{R}^{2d}} |\mathrm{Poly}(x,v)| \, (1+\mathcal{V}(x,v)) \, e^{-U(x;i)-|v|^2/2} \der x \der v < \infty.$$
\end{assumption}

We say that $\rE^*$ and $\rB^*$ are the formal adjoints of $\rE$ and $\rB$ with respect to the Lebesgue pairing on $\mathbb R^{2d}\times\mathcal M$ if, for every
$f\in C_c^\infty(\mathbb R^{2d}\times\mathcal M)$ and every smooth function $g$ for which the pairings below are absolutely integrable and integration by parts produces no boundary terms, we have
\begin{align}
\sum_{i\in\mathcal M}\int_{\mathbb R^{2d}} (\rE f)(x,v,i)\,g(x,v,i)\,\der x\der v
&=
\sum_{i\in\mathcal M}\int_{\mathbb R^{2d}} f(x,v,i)\,(\rE^* g)(x,v,i)\,\der x\der v,
\\
\sum_{i\in\mathcal M}\int_{\mathbb R^{2d}} (\rB f)(x,v,i)\,g(x,v,i)\,\der x\der v
&=
\sum_{i\in\mathcal M}\int_{\mathbb R^{2d}} f(x,v,i)\,(\rB^* g)(x,v,i)\,\der x\der v.
\end{align}
In particular, this applies to the functions $g$ arising later in the proof of Theorem~\ref{sebes_error_expansion_thrm}, such as $\rho$, $\rE^*\rho$, $\rB^*\rho$, and their finite compositions, since by Assumption~\ref{swi_decay_assump} these functions are smooth and have sufficiently fast decay.

\begin{theorem}[Error representation for $\mathrm{[SEBES]}$]\label{sebes_error_expansion_thrm}
Let Assumptions~\ref{assump_1_well_posed}-\ref{swi_decay_assump} hold. Let condition \eqref{condi_swi_hmc_lambda} be satisfied. Let $\mu_h$ denote the invariant measure of the Markov chain generated by the $\mathrm{[SEBES]}$ scheme. Then for every $\varphi \in L^{\infty}_{\cV}$ and sufficiently small $h$, we have
\begin{align}
\mu(\varphi)-\mu_h(\varphi)
=
h^2
\sum_{i\in\mathcal M}\int_{\mathbb R^{2d}}
\Psi(x,v,i)(\mathcal{S}^*\rho)(x,v,i)\der x\der v
+\mathcal{O}(h^3),
\end{align}
where $\Psi$ is a solution of the discrete Poisson equation \eqref{dis_Poisson_equation_SEBES} satisfying bound in \eqref{swi_hmc_poiss_bound}, and
\begin{align*}
\mathcal S^* &= \frac1{12}[\rB^*,[\rB^*,\rE^*]]
-\frac1{24}[\rE^*,[\rE^*,\rB^*]]
+\frac1{12}[\rB^*+\rE^*,[\rB^*+\rE^*,\rS^*]]
-\frac{1}{24}[\rS^*,[\rS^*,\rB^*+\rE^*]].
\end{align*}
\end{theorem}
For two probability measures $\nu_1,\nu_2$ on $\mathbb{R}^{2d}\times\mathcal{M}$,
define the weighted total variation distance associated with $\mathcal{V}$ by
\begin{align}
\|\nu_1-\nu_2\|_{\mathcal{V}}
:= \sup_{\|\varphi\|_{L^\infty_{\mathcal V}}\leq 1}
\Big|
\int \varphi\,\der\nu_1-\int \varphi\,\der\nu_2
\Big|.
\end{align}
We also define the usual total variation distance and the bounded-Lipschitz distance  by
\begin{align*}
\|\nu_1-\nu_2\|_{\mathrm{TV}}
:=
\sup_{\|\varphi\|_\infty\leq 1}
\big| \textstyle
\int \varphi\,\der\nu_1-\int \varphi\,\der\nu_2
\big|\; \text{ and }\;
d_{\mathrm{BL}}(\nu_1,\nu_2)
:=
\sup\limits_{\|\varphi\|_\infty\leq 1,\; \mathrm{Lip}(\varphi)\leq 1}
\left|
\int \varphi\,\der\nu_1-\int \varphi\,\der\nu_2
\right|,
\end{align*}
respectively.
\begin{corollary}
\label{cor_sebes_convergence}
Let the assumptions of Theorem~\ref{sebes_error_expansion_thrm} hold. Then there exist constants
$C>0$ and $h_0\in(0,1)$, independent of $h$, such that for all
$0<h\leq h_0$,
\begin{equation}
\|\mu-\mu_h\|_{\mathcal{V}}
\leq C h^2.
\end{equation}
Consequently, $
\|\mu-\mu_h\|_{\mathrm{TV}}
\leq C h^2$
and $d_{\mathrm{BL}}(\mu,\mu_h)
\leq C h^2$. 
Using Kantorovich-Rubinstein duality, we also have
\begin{align}
\mathcal{W}_1(\mu,\mu_h)\leq C h^2,
\end{align}
where $\mathcal{W}_1$  represents $1$-Wasserstein distance. 
\end{corollary}
The proof of above corollary follows from Theorem~\ref{sebes_error_expansion_thrm} since we have uniform in $h$ bound on the solution of discrete Poisson equation.

One approach to obtain  weak error expansions/bounds for ergodic estimates is to use the solution to a parabolic PDE:
\begin{equation*}
    \frac{\partial u}{\partial t} = \mathcal{L}u, \quad u(0, z) = \varphi(z),
\end{equation*}
where $\mathcal{L}$ is the operator associated with the underlying Markov process. This technique, based on backward Kolmogorov equation, has been employed in \cite{talay_tubaro_90,talayinv,debussche2012weak,abdulle2014ergo, tretyakov2025sampling} for hypoelliptic and elliptic stochastic differential equations assuming $U \in C^{\infty}$. This approach requires that the derivatives of the semigroup corresponding to SDE decay exponentially in time.   There is one other approach based on Poisson partial differential equation. It has been employed in \cite{mattingly_stuart_tretyakov, sharma2021,leimkuhler2023simplerandom, bharath2025sampling}. To get weak error bounds using this methodology, we require Schauder estimates on the solution of  Poisson equation. 

It is clear that both techniques,  which have been used for weak error analysis of numerical schemes to compute ergodic averages,  require certain types of bounds on the solutions of parabolic or elliptic PDEs.  The paper \cite{tretyakov2025sampling} investigating switching Langevin algorithm  conveniently relies on \cite{talayinv} for such bounds. In our case, we have a first order partial integro-differential equation (PIDE). We would need such bounds if we were to follow the above mentioned approaches for error analysis. Since we do not have elliptic or hypoelliptic diffusion in the dynamics, we are not aware of any such results  on the derivatives of the solution of corresponding first order PIDE. Therefore, we devise an approach based on discrete Poisson equation associated with the numerical scheme itself to obtain the error bounds. Discrete Poisson equation has been used as a tool for studying additive functionals, especially when proving limit theorems (see
\cite{benveniste2012adaptive,  nummelin1991poisson} and \cite[Chapter~17]{meyn2012markov}), and for convergence of stochastic approximation and reinforcement learning algorithms \cite{borkar2025ode, nanda2025minimal,chandak2022concentration}. One important consequence of the presented approach is that we get convergence in weighted total variation norm and $1$-Wasserstein distance (see Corollary~\ref{cor_sebes_convergence}) which is not straightforward in PDEs based approaches. 

For the discussion that follows let $\mathcal{L}$ denote the generator of kinetic Langevin dynamics. The error analysis for the numerical schemes of kinetic Langevin equation in \cite{lemakstoltz2016} also uses the operator $(\mathcal{P}_{h}- I)/h$ where $\mathcal{P}_h$ is the one-step Markov operator of the numerical scheme. However, the analysis  still relies on two Poisson PDEs : one involving generator $\mathcal{L}$ and other with its $L^2$ adjoint $\mathcal{L}^*$. In \cite{lu_xuda2025mean}, discrete Poisson equation with operator $(e^{\mathcal{L}h} - I)/h$ has been employed to obtain non-asymptotic bounds for  numerical schemes  approximating kinetic Langevin equation but assuming strongly convex $U$. In contrast, we  use  discrete Poisson equation associated with numerical scheme itself i.e.  we have $\mathcal{P}_h^{\EBE} = e^{h\mathrm{E}/2}e^{h\mathrm{B}}e^{h\mathrm{E}/2} $ and similarly for $\mathcal{P}_h^{\mathrm{SEBES}}$ to analyze the error. We leave it for future work to obtain non-asymptotic bounds on error incurred in invariant measure approximation  using numerical integrators for randomized HMC~\eqref{eqn_Poishmc}, switching HMC~\eqref{eqn_switching_Poishmc} or kinetic Langevin equation  with strongly convex potential $U$ (or $U_i$ in case with switching) using the discrete Poisson equation based approach proposed in this paper.

\section{Proofs}

\subsection{Proof of Theorem~\ref{sebes_geom_erg_thrm}}

\begin{proof}[Proof of Lemma~\ref{lemma_uniform_minor_sebes}]

For $0< h \leq h_0$, set $
t^{\varepsilon, h} :=K_\varepsilon h$. Then $\varepsilon\leq t^{\varepsilon, h} <\varepsilon+ h$.  We will first prove minorization for $[\EBE]$ scheme. Then, we will provide arguments for minorization for $[\mathrm{SEBES}]$. We choose $h_0 \leq \varepsilon/12$. We consider the intervals
\begin{equation}
I_1:=\left[\frac{\varepsilon}{6},\frac{2\varepsilon}{6}\right],
\qquad
I_2:=\left[\frac{4\varepsilon}{6},\frac{5\varepsilon}{6}\right]
\end{equation}
which are contained in $[0,t^{\varepsilon, h}]$, and for $
\tau_1\in I_1$, $\tau_2\in I_2$, 
we have 
\begin{equation}
\tau_2-\tau_1\geq \frac{\varepsilon}{3}.
\end{equation}
Let $\mathcal E_{\varepsilon,h}$ be the event that there is one refreshment
in \(I_1\), one refreshment in $I_2$, and no other refreshment in the interval
$[0,t^{\varepsilon, h}]$. Due to the independence of increments of the Poisson process, we have
\begin{equation}
\mathbb P(\mathcal E_{\varepsilon,h})
=
e^{-\lambda t^{\varepsilon, h}}
\left(\lambda\frac{\varepsilon}{6}\right)^2 \geq
e^{-\lambda(\varepsilon+ h_0)}
\left(\lambda\frac{\varepsilon}{6}\right)^2
=:p_\varepsilon>0.
\end{equation}
  Conditioned on $\mathcal E_{\varepsilon,h}$, consider the two refreshment times $
\tau_1\in I_1$ and $ \tau_2\in I_2$ and let $Z_1,Z_2\in\mathbb{R}^d$ be the corresponding independent standard Gaussian variables.
For fixed $(x_0,v_0)$, $\tau_1,\tau_2$, and $h$, let us write the final state of $\EBE$ recursion as a function of $(Z_1, Z_2)$ 
\begin{equation}
(x_{K_\varepsilon},v_{K_\varepsilon})
=
\Phi(x_0,v_0,\tau_1,\tau_2;Z_1,Z_2).
\end{equation}
We similarly denote by
\begin{equation}
\Phi^{0}(x_0,v_0,\tau_1,\tau_2;Z_1,Z_2)
\end{equation}
the corresponding zero force  map i.e. the final state of $\EBE$ recursion with $U\equiv 0$ while keeping the same initial condition, the same refreshment times, and the same $(Z_1,Z_2)$. The idea of treating $\Phi$ as a perturbation of $\Phi_0$ comes from \cite{lemakstoltz2016}. The major difference is that we do not assume positional space to be compact. 
However, the fundamental idea remains the same, like in several other minorization proofs, that is to find bounds on the sensitivity of $\Phi$ (see \eqref{swi_hmc_det_lowb_Phi} and \eqref{swi_hmc_new_eqn_5.10} below) with respect to noise variable so that change of variable formula can be applied appropriately to obtain required lower bound for minorization.

The proof rests on the following two claims which we prove below:
\begin{itemize}[leftmargin=*, labelwidth=4.2em, labelsep=0.5em, align=left]
    \item[Claim 1:] For a $z_* \in \mathbb{R}^{2d}$, the equation $\Phi^0(x_0, v_0, \tau_1, \tau_2 ; Z^0) =z_* $ has a unique solution $Z^0 = (Z_1^0, Z_2^0) $. Moreover,  there exists an $R_0$ satisfying $|Z^0| \leq R_0$  uniformly over  $(x_0,v_0) \in \mathcal{C}$, $\tau_1,\tau_2$, and
$0<h\leq h_0$. More precisely, $R_0 = C_0/\varepsilon  $ with $C_0> 0$ being independent of $h$.  
   
\item[Claim 2]   Let $R:=2R_0$. We  choose $\varepsilon>0$ sufficiently small so that there exists a constant  $b_0>0$ such that
\begin{equation}\label{swi_hmc_det_lowb_Phi}
\left| \det D_Z\Phi(x_0,v_0,\tau_1,\tau_2;Z)\right| \geq b_0
\end{equation}
for all $(x_0,v_0)\in \mathcal{C}$, $\tau_1\in I_1$, $\tau_2\in I_2$, $|Z|\leq R$, $0< h \leq h_0$ with $Z = (Z_1, Z_2)$.
Moreover, $D_Z\Phi$ satisfies
\begin{equation}\label{swi_hmc_new_eqn_5.10}
\|D_Z\Phi(x_0,v_0,\tau_1,\tau_2;Z)
- J^{0}\| \leq C(\varepsilon+h_0)^2
\end{equation}
on $\mathcal C\times I_1\times I_2
\times \{ Z=(Z_1,Z_2)\in\mathbb R^{2d}:
|Z|\leq R \} $ and $J^{0}$ is from \eqref{swi_eqn_J_0}.
\end{itemize}
\textit{Proof of Claim 1:} For fixed $(x_0,v_0)$, \(\tau_1,\tau_2\), and $h$, the $\Phi^0$ is affine
in the two Gaussian variables, i.e.
\begin{equation}
\Phi^{0}(x_0,v_0,\tau_1,\tau_2;Z_1,Z_2)
=
M^{0}(x_0,v_0,\tau_1,\tau_2)
+
J^{0}(\tau_1,\tau_2)(Z_1, Z_2)^{\top},
\end{equation}
where 
\begin{equation}
M^0 = \begin{pmatrix}
x + v( \tau_1 + \cos(\phi)(\tau_2-\tau_1) + (\cos(\phi))^2(t^{\varepsilon,h}-\tau_2)) \\[4pt]
(\cos(\phi))^2v
\end{pmatrix},
\end{equation}
and
\begin{equation}\label{swi_eqn_J_0}
J^{0} =  \sin(\phi)
\begin{pmatrix}
\big((\tau_2-\tau_1)+\cos(\phi)(t^{\varepsilon, h}-\tau_2)\big)I_d & (t^{\varepsilon, h}-\tau_2)I_d
\\ \cos(\phi)I_d & I_d
\end{pmatrix}.
\end{equation}
Since $ \tau_2-\tau_1\geq \frac{\varepsilon}{3}$ and $\det J^{0}
=
\sin(\phi)^{2d}(\tau_2-\tau_1)^d$, we get
\begin{equation}
|\det J^{0}|
\ge
\sin(\phi)^{2d}\left(\frac{\varepsilon}{3}\right)^d>0.
\end{equation}
Moreover, since $
0\leq t^{\varepsilon, h} - \tau_2\le  t^{\varepsilon, h} \leq \varepsilon+ h_0 \leq \frac{13}{12}\varepsilon$, 
there exists $C_{\phi}<\infty$, independent of $h$, such that
\begin{equation} \label{upper_bound_J0-1_eqn}
    \|(J^{0})^{-1}\|\leq C_{\phi}/\varepsilon.
\end{equation}
 Since \(\mathcal{C}\) is compact,
there exists $M_{\max}<\infty$, independent of $h$, such that
$ |M^{0}(x_0,v_0,\tau_1,\tau_2)|\leq M_{\max}$. Therefore, there exists a constant $C_d >0$ so that
\begin{equation}\label{upper_bound_Z_0_claim_1}
|Z^0| \leq
C_d C_{\phi}\big(|z_*|+M_{\max}\big)/\varepsilon.
\end{equation}
\textit{Proof of Claim 2:} 
 Let $y_n =(x_n,v_n)^\top$ and $y_n^0=(x_n^0,v_n^0)^\top$ denote the nonzero-force and zero-force $\mathrm{EBE}$ iterates, respectively.
Conditioned on $\mathcal{E}_{\varepsilon, h}$ and two refreshment times $\tau_1$ and \(\tau_2\), two Gaussian variables denoted as $Z=(Z_1,Z_2)\in\mathbb R^{2d}$ enter the
recursion at the refreshment
times $\tau_1$ and $\tau_2$. We write the corresponding $\EBE$ recursion below:
\begin{align*}
    \hat y_{k-1/2} &= e_{k, 1} y_{k-1} +  G_{k,1}(Z_1, Z_{2}), \quad
    y_{k-1/2} = \hat{y}_{k-1/2} + h (0, -\nabla U(x_{k-1/2}))^{\top}, \\
y_{k} &= e_{k,2}y_{k-1/2} + G_{k,2}(Z_{1}, Z_2),
\end{align*}
where $e_{k,i}$, $i=1,2$ denote the deterministic linear part of E flow, and $G_{k,1}(\cdot, \cdot)$ and $G_{k,2}(\cdot, \cdot)$ depend on $Z_1, Z_2$. For all half-steps not containing $\tau_1$ or $\tau_2$, the corresponding $G_{k,i}$, $i=1,2$ is zero. Therefore, by repeatedly applying the recursion above, we can obtain the following  for every
$1\le n\le K_\varepsilon$:
\begin{align}\label{swi_hmc_ebe_recur_minor_proof}
y_n = y_n^0 - h\sum_{k=1}^n
P_{n,k} (0, \nabla U(x_{k-1/2}))^{\top},
\end{align}
where $P_{n,k} = w_n w_{n-1}\cdots w_{k+1} e_{k,2}$ with $w_j = e_{j,2}e_{j,1}$ and $P_{n,n} = e_{n,2}$, and $P_{n,k}$ does not depend on $y_k$, therefore,
there exists a constant $C$, independent of $h$, such that
\begin{align}
\|P_{n,k}\|\le C,
\qquad 1\leq k\leq n\leq K_\varepsilon .
\end{align}
We differentiate \eqref{swi_hmc_ebe_recur_minor_proof} with respect to
$Z=(Z_1,Z_2)$. Since $P_{n,k}$ is independent of $Z$ and
$U\in C^2$, we obtain
\begin{equation}
D_Z y_n = D_Z y_n^0 - h\sum_{k=1}^n P_{n,k}
\begin{pmatrix}
0\\
\nabla^2 U(x_{k-1/2})D_Zx_{k-1/2}
\end{pmatrix}.
\end{equation}
Let $\tilde{Y}_n:=D_Z y_n - D_Z y_n^0$
and define analogously $\tilde{Y}_{k-1/2}:= D_Z y_{k-1/2}-D_Z y_{k-1/2}^0$. Then
\begin{equation}
D_Zx_{k-1/2} = D_Zx_{k-1/2}^0+(\tilde{Y}_{k-1/2})_x,
\end{equation}
where $(\tilde{Y}_{k-1/2})_x$ denotes $x$-component of $\tilde{Y}_{k-1/2}$. Hence
\begin{align}\label{swi_hmc_sensitivity_eqn}
\tilde{Y}_n = - h\sum_{k=1}^n P_{n,k}
\begin{pmatrix}
0\\
\nabla^2 U(x_{k-1/2})
\bigl(D_Zx_{k-1/2}^0+(\tilde{Y}_{k-1/2})_x\bigr)
\end{pmatrix}.
\end{align}
For steps when $t_{k-1/2} \leq \tau_1$, the noise has not entered the iterates and therefore, $D_Z x_{k-1/2}^0 = 0$.  For $t_{k-1/2} > \tau_1$, the position sensitivity grows linearly in the remaining time, i.e. 
\begin{equation}\label{hmc_swi_eqn_5.21}
\|D_Z x_{k-1/2}^0\| \le C (t_{k-1/2} - \tau_1)_+,
\end{equation}
where  $C > 0$ is independent of $h$. Taking norms on both sides in \eqref{swi_hmc_sensitivity_eqn} and using \eqref{hmc_swi_eqn_5.21} yield
\begin{align*}
\|\tilde{Y}_n \| &\leq  C h\sum_{k=1}^n
\Big( (t_{k-1/2}-\tau_1)_+
+ \|\tilde{Y}_{k-1/2}\|\Big) 
\leq C h\sum_{k=1}^n
(t_{k-1/2}-\tau_1)_+
+
C h\sum_{k=1}^n
\max_{0\le j\le k}\|\tilde{Y}_j\|,
\end{align*}
where $C$ is independent of $h$. Since $h\sum_{k=1}^n (t_{k-1/2}-\tau_1)_+ \leq  C(t^{\varepsilon,h}-\tau_1)^2 \leq C(\varepsilon+h_0)^2 $, we obtain
\begin{align}
\max_{1 \leq k \leq n} \|\tilde{Y}_k \| \leq C (\varepsilon + h_0)^2 + C h\sum_{k=0}^{n-1} \max_{0\le j\le k}\|\tilde{Y}_j\|.
\end{align}
Using the discrete Gronwall lemma, we get $\max_{1 \leq k \leq n} \|\tilde{Y}_k \|
\leq C(\varepsilon+h_0)^2
e^{CK_\varepsilon h}$.
As we know $K_\varepsilon h=t^{\varepsilon,h}\le \varepsilon+h_0$, we ascertain
\begin{equation}\label{new_eqn_swi_hmc_5.20}
\|D_Z\Phi-J^0\| = \|\tilde{Y}_{K_\varepsilon}\| \leq C(\varepsilon+h_0)^2.
\end{equation}
It remains to obtain \eqref{swi_hmc_det_lowb_Phi}. We have already proved that $\|(J^0)^{-1}\|\leq \frac{C_\phi}{\varepsilon}$. Hence
\begin{align}
\|(J^0)^{-1}\|\,\|D_Z\Phi-J^0\|
\leq \frac{C_\phi}{\varepsilon}C(\varepsilon+h_0)^2.
\end{align}
Choosing $\varepsilon>0$ sufficiently small, we can ensure that
\begin{equation}\label{swi_hmc_eqn_new_5.24}
\|(J^0)^{-1}\|\,\|D_Z\Phi-J^0\|\leq \frac{1}{2} .
\end{equation}
We write $D_Z\Phi = J^0[I+(J^0)^{-1}(D_Z\Phi-J^0)]$ and thanks to \eqref{swi_hmc_eqn_new_5.24}, every singular value of $[I+(J^0)^{-1}(D_Z\Phi-J^0)]$ is bounded below by $1/2$. Therefore, $|\det D_Z\Phi| \geq 2^{-2d}|\det J^0|$ and using $|\det J^0| = \sin(\phi)^{2d}(\tau_2-\tau_1)^d \geq \sin(\phi)^{2d}\left(\frac{\varepsilon}{3}\right)^d$, we obtain
\begin{equation}
|\det D_Z\Phi|
\geq 2^{-2d}\sin(\phi)^{2d}
\left(\frac{\varepsilon}{3}\right)^d =:b_0>0.
\end{equation}
The constant $b_0$ is independent of $h$, $(x_0,v_0)$, $\tau_1$, $\tau_2$,
and $Z$ with $|Z|  \leq R$. This completes the proof of Claim~2.

Recall $R =2R_0$ and $R_0:= C_0/\varepsilon$. Let $r_1:=R_0$. Then,
$B(Z^0,r_1) \subset \{ Z=(Z_1,Z_2)\in\mathbb R^{2d}:\ |Z|\le R\}$. 
Consider the map $Z\mapsto \Phi(x_0,v_0,\tau_1,\tau_2;Z) $ on the ball $B(Z^0,r_1)$.  Due to \eqref{swi_hmc_det_lowb_Phi} and \eqref{new_eqn_swi_hmc_5.20}, there exists a constant $\hat{K}$ uniformly in
$(x_0,v_0)\in\mathcal C$, \(\tau_1\in I_1\), \(\tau_2\in I_2\), and
\(0<h\le h_0\), such that
\begin{equation}
\|(D_Z\Phi(x_0,v_0,\tau_1,\tau_2;Z^0))^{-1}\|
\leq \frac{\hat{K}}{\varepsilon}. \label{hmc_swi_eqn_5.28}
\end{equation}
Moreover, since \(J^0\) is independent of $Z$, for every
\(Z\in B(Z^0,r_1)\), using \eqref{swi_hmc_new_eqn_5.10}, we get
\begin{align}
&\|D_Z\Phi(x_0,v_0,\tau_1,\tau_2;Z)
- D_Z\Phi(x_0,v_0,\tau_1,\tau_2;Z^0)
\| \nonumber  \\ & \quad\leq \|D_Z\Phi(x_0,v_0,\tau_1,\tau_2;Z)-J^0\|
+ \|D_Z\Phi(x_0,v_0,\tau_1,\tau_2;Z^0)-J^0\| \leq
2C(\varepsilon+h_0)^2. \label{hmc_swi_eqn_5.29}
\end{align}
We choose  $\varepsilon>0$ and then $h_0>0$ sufficiently small so that the following bound holds:
\begin{equation}\label{swi_eqn_eqn_5.25}
2C(\varepsilon+h_0)^2\le \frac{\ell \varepsilon}{\hat{K}},
\qquad 0<\ell<1,
\end{equation}
where $C$ is from \eqref{hmc_swi_eqn_5.29} and $\hat{K}$ is from \eqref{hmc_swi_eqn_5.28}. Therefore, based on \eqref{hmc_swi_eqn_5.28} and  \eqref{hmc_swi_eqn_5.29} combined with \eqref{swi_eqn_eqn_5.25}, we can apply quantitative inverse function theorem \cite[Chapter~XIV, Theorem~1.2 and Lemma~1.3]{lang_real}  (also see \cite[Proposition~2.5.6]{abraham1988manifolds}) over all admissible parameters, to get
\begin{align}
B\left(\Phi(x_0,v_0,\tau_1,\tau_2;Z^0),r_0\right)
\subset
\Phi(x_0,v_0,\tau_1,\tau_2;B(Z^0,r_1)),
\end{align}
where $r_0:=(1-\ell)r_1 \varepsilon/\hat{K}>0$.  Since $r_1 = R_0 = C_0/\varepsilon$, we have $r_0 = (1- \ell) C_0 /\hat{K} > 0$. 

Although $|Z^0|\leq C_0/\varepsilon$, the chain $x^0_k$ evolves in a bounded domain, since the contribution of the refreshment variables $Z^0$ to the position is multiplied by time interval of
length $\mathcal{O}(\varepsilon)$. Then using \eqref{swi_hmc_ebe_recur_minor_proof},  using the  global Lipschitzness of $\nabla U$ and finally applying discrete Gronwall argument over the time interval $t^{\varepsilon , h} \leq \varepsilon + h_0$, it is not difficult to show that $x_k$ also belongs to a compact set which size is independent of $\varepsilon$.  Therefore, from \eqref{swi_hmc_ebe_recur_minor_proof}, we obtain 
\begin{align}
\left|
\Phi(x_0,v_0,\tau_1,\tau_2;Z^0)-z_*
\right|
&\leq h\sum_{k=1}^{K_\varepsilon}
\|P_{K_\varepsilon,k}\|\,
|\nabla U(x_{k-1/2})| \leq
C h K_{\varepsilon}.
\end{align}
 We choose sufficiently small $\varepsilon$ such that   \eqref{swi_eqn_eqn_5.25} holds as well as the following bound holds:
\begin{align}
\left|\Phi(x_0,v_0,\tau_1,\tau_2;Z^0)-z_*\right|
\leq C(\varepsilon+h_0) \leq \frac{r_0}{2}.
\end{align}
Thus, denoting $D:=B(z_*,\frac{r_0}{2})$,
we have $D\subset
B\left(\Phi(x_0,v_0,\tau_1,\tau_2;Z^0),r_0\right)
\subset
\Phi(x_0,v_0,\tau_1,\tau_2;B(Z^0,r_1)).$
Consequently, every $z\in D$ has a pre-image $Z(z)\in B(Z^0,r_1)$. In particular, $|Z(z)|\leq R$. By the change of variables formula, the conditional density of
$(X_{K_\varepsilon},V_{K_\varepsilon})$, given
$\mathcal E_{\varepsilon,h}$ and $(\tau_1,\tau_2)$, satisfies, for
$z\in D$,
\begin{align}
f(z\mid x_0,v_0,\tau_1,\tau_2)
\geq \frac{G_{2d}(Z(z))}{ \left|
\det D_Z\Phi(x_0,v_0,\tau_1,\tau_2;Z(z))
\right|}.
\end{align}
Since $|Z(z)|\leq R$, the standard Gaussian density on
$\mathbb R^{2d}$ satisfies $ G_{2d}(Z(z)) \geq (2\pi)^{-d}e^{-R^2/2}$. 
Moreover, using \eqref{swi_hmc_new_eqn_5.10} and the uniform boundedness of $J^0$, there exists
$\tilde{C}<\infty$, independent of $h$, $(x_0,v_0)$, $\tau_1$, and
$\tau_2$, such that $
| \det D_Z\Phi(x_0,v_0,\tau_1,\tau_2;Z) |
\leq \tilde{C} $
for all \(Z\in B(Z^0,r_1)\). Therefore, $
f(z\mid x_0,v_0,\tau_1,\tau_2) \geq (2\pi)^{-d}e^{-R^2/2}\tilde{C}^{-1} =:\delta>0$, $z\in D$. 
The constant \(\delta\) is independent of \((x_0,v_0)\in\mathcal C\),
\((\tau_1,\tau_2)\in I_1\times I_2\), and \(0<h\le h_0\). Integrating over the event $\mathcal E_{\varepsilon,h}$, whose probability is bounded below by $p_\varepsilon>0$, we obtain, for every measurable $O\subset\mathbb R^{2d}$,
\begin{align}
(\mathcal{P}_h^{\EBE})^{K_\varepsilon}((x,v),O)
\geq p_\varepsilon\,\delta\,\operatorname{Leb}(D\cap O) = \alpha \nu(O),
\end{align}
where $\nu(O):=
\frac{\operatorname{Leb}(D\cap O)}{\operatorname{Leb}(D)}$ and $\alpha:=p_\varepsilon\delta\operatorname{Leb}(D)>0$.

We now extend the preceding set of arguments to the $\mathrm{SEBES}$ scheme. Conditional on a realization of switching variables, the $\EBE$ recursion remains unchanged except the force in $n$-th $\B$ step is $\nabla U(\cdot, \s_{n +1/2})$. Therefore, thanks to our Assumptions~\ref{assump_1_well_posed}-\ref{assump_2_ergo} which are uniform in switching variable,  Claim~1 and Claim~2 hold true uniformly in switching variable. In particular, we can choose $R$, $\delta$ and $D$ independently of switching path. Conditional on $|Z| \leq R$, the corresponding trajectories remain in a compact set $\mathcal{K}_R$, uniformly over $(x_0, v_0) \in \mathcal{C}$, $0 \leq h \leq h_0$ and all switching realizations. Let $\mathcal{K}_{x, R} := \{ x\in \mathbb{R}^{d}\, ;\,  (x,v) \in \mathcal{K}_R\}$.   Due to our assumptions on $\q_{ij}$, there exist constants $\q_* = \inf_{x \in \mathcal{K}_{x, R}} \min_{i\neq j} \q_{ij}(x) > 0$ and $\q^{*} := \sup_{x \in \mathcal{K}_{x, R}} \max_{i \in \mathcal{M}} \q_{i}(x) < \infty$. Fix $\hat{k} \in \mathcal{M}$. If $i =\hat{k}$, the event that no switching jump occurs on $[0, t^{\varepsilon,h}]$ has probability at least $e^{-\q^{*}(\varepsilon + h_0)}$. If $i \neq \hat{k}$, consider the event that the switching chain makes exactly one jump from $i$ to $\hat{k}$ during interval $\varepsilon/2$ and makes no other jump during the whole interval $[0 , t^{\varepsilon, h}]$. Since $\q_{i\hat{k}}(x) \geq \q_*$ on $\mathcal{K}_{x, R}$ and the total jump rate is bounded above by $\q^*$, the probability of the concerned event is lower bounded by $e^{-\q^* t^{\varepsilon, h}}\q_* \varepsilon/2 \geq e^{-\q^*( \varepsilon + h_0)}\q_* \varepsilon/2$. As a consequence, we obtain 
\begin{align}
    \mathbb{P}(\s_{K_\varepsilon} = \hat{k} \mid \s_0 = i) \geq p_s, \quad i \in \mathcal{M},
\end{align}
where $ p_s = e^{-\q^*(\varepsilon + h_0)} \min\{ 1, \q_* \varepsilon/2\} >0$. Therefore, for every measurable $\hat{O} \subset \mathbb{R}^{2d} \times \mathcal{M}$ and every $(x, v, i) \in \mathcal{C} \times \mathcal{M}$, we obtain
\begin{align}
(\mathcal{P}^{\mathrm{SEBES}})^{K_{\varepsilon}} \geq p_\varepsilon \delta p_S\sum_{k\in\mathcal M}
\operatorname{Leb}\big(\{z\in D:(z,k)\in \hat{O}\}\big).
\end{align}
Defining the probability measure
\begin{align}
\nu(\hat{O})
:=
\frac{1}{|\mathcal M|\,\operatorname{Leb}(D)}
\sum_{k\in\mathcal M}
\operatorname{Leb}\big(\{z\in D:(z,k)\in \hat{O}\}\big),
\end{align}
we have
\begin{align}
(\mathcal{P}^{\mathrm{SEBES}}_h)^{K_{\varepsilon}}\big((x,v,i),\hat{O}\big)
\geq
\alpha \nu(\hat{O}),
\end{align}
where  $\alpha := p_\varepsilon\delta p_S\,\mathrm{Leb}(D)|\mathcal M|>0$.  
\end{proof}

\begin{proof}[Proof of Lemma~\ref{lemma_lyapunov_ebe}]
Let $\mathcal{F}$ be the sigma-algebra generated by the Poisson jump process over the entire interval $[0,h]$. Since the $\mathrm{E}$ steps are of size $h/2$, we split the information into two independent sigma-algebras. Let $\mathcal{F}_1$ be generated by the Poisson process on $[0,h/2)$, with $N_1=N_{h/2}=n_1$ and jump times $\{\tau^{(1)}\}$. Let $\mathcal{F}_2$ be generated by the Poisson process on $[h/2,h)$, with $N_2=N_h-N_{h/2}=n_2$ and jump times $\{\tau^{(2)}\}$. Here $\{\tau^{(1)}\} = \{\tau_1^{(1)}, \dots, \tau_{n_{1}}^{(1)}\}$ denotes the jump times in $[0,h/2)$ conditional on $N_{h/2}=n_1$, while $\{\tau^{(2)}\}$ denotes the jump times in $[h/2,h)$ conditional on $N_h-N_{h/2}=n_2$.

For the sake of convenience, we denote the random coefficients for the first and second $\mathrm{E}$ steps conditional on $\mathcal{F}_1$ and $\mathcal{F}_2$:
 $A_1 = \sum_{i=0}^{n_1} \cos^i(\phi)(\tau_{i+1}^{(1)} - \tau_i^{(1)})$ and $B_1 = \cos^{n_1}(\phi)$, and $A_2 = \sum_{j=0}^{n_2} \cos^j(\phi)(\tau_{j+1}^{(2)} - \tau_j^{(2)})$ and $B_2 = \cos^{n_2}(\phi)$. Note that $\cos(\phi)\geq 0$, $0 <A_i \leq h/2$ and $0 \leq B_i \leq 1$ are $\mathcal{F}_i$-measurable, $i=1,2$ random variables.  Recall the intermediate transitions are
 $(x_{1/2}, \hat{v}_{1/2}) = \mathscr{S}^\mathrm{E}_{h/2}(x_0, v_0)$,  $(x_{1/2}, v_{1/2}) = \mathscr{S}^\mathrm{B}_{h}(x_{1/2}, \hat{v}_{1/2})$ and $(x_{1}, v_{1}) = \mathscr{S}^\mathrm{E}_{h/2}(x_{1/2}, v_{1/2})$.
We can write $(x_1, v_1)$ as
\begin{align} 
(x_1, v_1)^{\top} = (x_{1/2} +v_{1/2}A_2, v_{1/2}B_2)^{\top} + J(\xi_1^{(2)}, \xi_2^{(2)})^{\top},
\end{align}
where $J$ is \begin{equation}\label{swi_equation_5.48}
    J = \begin{pmatrix} 
    \sin(\phi) \frac{\sigma_{xv}}{\sqrt{\sigma_{vv}}} I_{d} & \sin(\phi) \sqrt{\sigma_{xx} - \frac{\sigma_{xv}^2}{\sigma_{vv}}} I_{d} \\[8pt]
    \sin(\phi)\sqrt{\sigma_{vv}} I_{d} & \mathbf{0}_{d}
    \end{pmatrix}  = \sin(\phi) \begin{pmatrix}  \frac{\sigma_{xv}}{\sqrt{\sigma_{vv}}} &  \sqrt{\sigma_{xx} - \frac{\sigma_{xv}^2}{\sigma_{vv}}}  \\[8pt]
    \sqrt{\sigma_{vv}} & 0\end{pmatrix} \otimes I_{d}.
\end{equation} 
If $n=0$ then $J = \mathbf{0}_d$. 
It is not difficult to obtain that $\operatorname{Trace}(JJ^\top) = d \sin^2 (\phi) ( \sigma_{xx} + \sigma_{vv})$, and therefore, $\mathbb{E}(\operatorname{Trace}(JJ^{\top})) \leq C h$ where $C$ is independent of $h$.  

Recall $W (x,v)  = 1 + \mathcal{V}(x,v) = 1 + a|x|^2 + b|v|^2 + c( x \cdot v) $ with $ab > c^2/4$.  Conditioned on $\mathcal{F}$ and $(x_{1/2}, v_{1/2})$, we get
\begin{align*}
\mathbb{E}[\mathcal{V}(x_1, v_1) &\mid x_{1/2}, v_{1/2}, \mathcal{F}] = \mathbb{E}[a|x_1|^2 + b|v_1|^2 + c(x_1 \cdot v_1)  \mid x_{1/2}, v_{1/2}, \mathcal{F}] \\  
&=   \mathbb{E}\big[ a | x_{1/2} +v_{1/2} A_2 + J_{11} \xi^{(2)}_1 + J_{12} \xi^{(2)}_2 |^2 + b | v_{1/2}B_2 + J_{21} \xi_{1}^{(2)} |^2 \\
&\quad + c\big((x_{1/2} +v_{1/2} A_2 + J_{11} \xi^{(2)}_1 + J_{12} \xi^{(2)}_2) \cdot (v_{1/2}B_2 + J_{21} \xi_{1}^{(2)})\big) \mid x_{1/2}, v_{1/2}, \mathcal{F} \big] \\
&= a |x_{1/2}|^2 + (aA_2^2 + b B_2^2 + cA_2 B_2) |v_{1/2}|^2 + (2 a A_2 + cB_2) (x_{1/2} \cdot v_{1/2}) + C_2,
\end{align*}
where $C_2$ is an $\mathcal{F}_2$-measurable random variable with all moments bounded and is independent of $\mathcal{F}_1$. Indeed, $C_2$ consists of second moments and cross moments of the Gaussian random vector used in the second $\mathrm{E}$ half-step. All cross terms are bounded by the corresponding second moments, and therefore $\mathbb{E}[C_2]\leq C \mathbb{E}[\operatorname{Trace}(JJ^\top)]\leq C h$.

For brevity, denote $K_1 := (aA_2^2 + bB_2^2+ cA_2 B_2)$ and $K_2 := (2aA_2 + cB_2)$.  We have
\begin{align*}
|x_{1/2}|^2 &= |x_0 + v_0 A_1 + J_{11} \xi_1^{(1)} + J_{12}\xi_2^{(1)}|^2, \quad
|\hat{v}_{1/2}|^2 = |v_0 B_1 + J_{21} \xi_{1}^{(1)}|^2, \\
(x_{1/2} \cdot \hat{v}_{1/2}) &= \big( (x_0 + v_0 A_1 + J_{11} \xi_1^{(1)} + J_{12}\xi_2^{(1)} ) \cdot (v_0 B_1 + J_{21} \xi_{1}^{(1)})\big).
\end{align*}
Also, denote $K_3 = a + L_1K_1 h^2 + L_1\Upsilon^{-1}K_1h - K_2 h \kappa  $. Therefore, applying Young's inequality gives
\begin{align}\label{swi_equation_5.16}
    \mathbb{E}[\mathcal{V}(x_1, v_1) \mid x_{0}, v_0, \mathcal{F}] &\leq  K_3 |x_0|^2   +  \big(K_3 A_1^2 + K_1(1 +  h\Upsilon)B_1^2 + K_2A_1 B_1\big) |v_0|^2  \nonumber \\  &
    +  \big( 2K_3 A_1  + B_1 K_2\big)(x_0 \cdot v_0) + C_1,
\end{align}
where $C_1$ is a generic $\mathcal{F}$ measurable random variable with all moments bounded. Moreover, $\mathbb{E}(C_1) \leq C h$ as $C_1$ is obtained from the Gaussian noise contributions in the two $\mathrm E$ half-steps. We now take expectation in \eqref{swi_equation_5.16} to get
\begin{align*}
\mathbb{E}[\mathcal{V}(x_1,v_1)\mid x_0,v_0]
\leq \alpha|x_0|^2 + \beta|v_0|^2 + \delta(x_0\cdot v_0) + C_0,
\end{align*}
where we have denoted
  $\alpha := \mathbb{E}[K_3]$, $ 
\beta := \mathbb{E}\big[K_3 A_1^2 + K_1(1+h\Upsilon)B_1^2 + K_2 A_1 B_1\big]$, $
\delta := \mathbb{E}\big[2 K_3 A_1 + B_1 K_2\big] $ 
and \(C_0 := \mathbb{E}[C_1]\).

Let us look at the coefficients of $|x_0|^2$, $|v_0|^2$ and $(x_0 \cdot v_0)$ one by one.   To this end, we will be using the facts which follow : we have $\mathbb{E}[B_i] = 1 - \frac{\lambda}{2}(1 - \cos(\phi))h + \mathcal{O}(h^2)$, $i=1,2$, 
$\mathbb{E}[B_1^2 B_2^2] = \exp\left(-\lambda h \sin^2(\phi)\right) = 1 - (\lambda \sin^2(\phi))h + \mathcal{O}(h^2)$ and $\mathbb{E}[B_1 B_2] = \exp\left(-\lambda h (1 - \cos(\phi))\right) = 1 - \lambda(1 - \cos(\phi))h + \mathcal{O}(h^2)$. We can compute for $\phi \neq \pi/2$
\begin{equation}
\mathbb{E}[A_1B_1] = \frac{
e^{-\lambda(1-\cos\phi)h/2} - e^{-\lambda(1-\cos^2\phi)h/2} }{ \lambda\cos\phi(1-\cos\phi)},
\end{equation}
implying $\mathbb{E}[A_1B_1] = h/2 + \mathcal{O}(h^2)$ for sufficiently small $h$. We also have $ \mathbb{E}[A_i]
=
\frac{1-e^{-\lambda(1-\cos\phi)h/2}}
{\lambda(1-\cos\phi)}$. 
Therefore, $
\mathbb{E}(A_i) = \frac{h}{2} + \mathcal{O}(h^2)$ and $ \mathbb{E}(A_i^2) = \mathcal{O}(h^2)$ for $\phi \neq \pi/2$. The same asymptotics can be obtained for $\phi = \pi/2$. Substituting for $K_1$ and $K_2$, we obtain
\begin{align}
\alpha = \mathbb{E}[K_3] = a + L_1h^2\mathbb{E}[K_1] + L_1 \Upsilon^{-1}h\mathbb{E}[K_1]  - \kappa h\mathbb{E}[K_2] = a - (c\kappa  - b L_1 \Upsilon^{-1}) h + \mathcal{O}(h^2).
\end{align}
Next, for the coefficient of \(|v_0|^2\), using $
\mathbb{E}[K_3A_1^2] = \mathcal{O}(h^2)$, and $
\mathbb{E}[K_2A_1B_1]
= \mathbb{E}[(2aA_2+cB_2)A_1B_1] 
= 2a\,\mathbb{E}[A_2]\mathbb{E}[A_1B_1]
+ c\,\mathbb{E}[B_2]\mathbb{E}[A_1B_1] 
= \frac{c}{2}h + \mathcal{O}(h^2)
$, 
we get
\begin{align}
\beta
&= \mathbb{E}\big[K_3 A_1^2 + K_1(1+h\Upsilon)B_1^2 + K_2A_1 B_1\big] \nonumber\\
&= (1+h\Upsilon )\mathbb{E}\big[aA_2^2B_1^2 + bB_2^2B_1^2 + cA_2B_2B_1^2\big] + \mathbb{E}[(2aA_2+cB_2)A_1B_1]\nonumber\\
&= b(1+h\Upsilon)\mathbb{E}[B_1^2B_2^2]
+ c\,\mathbb{E}[A_2B_2B_1^2] +2a\,\mathbb{E}[A_2]\mathbb{E}[A_1B_1]
+ c\,\mathbb{E}[B_2]\mathbb{E}[A_1B_1]   + \mathcal{O}(h^2) \nonumber\\
&= b\left(1+ \Upsilon h\right)\left(1-\lambda\sin^2(\phi)h+ \mathcal{O}(h^2)\right)
+ \frac{c}{2}h + \frac{c}{2}h + \mathcal{O}(h^2) \nonumber\\
&= b - b\left(\lambda\sin^2(\phi)-\Upsilon\right)h
+ ch + \mathcal{O}(h^2).
\end{align} 
Finally, for the coefficient of $(x_0\cdot v_0)$, we use again that $A_1=\mathcal{O}(h)$, $A_2=\mathcal{O}(h)$ to obtain
\begin{align}
\delta
&= \mathbb{E}[2K_3A_1 + B_1K_2] = 2a\,\mathbb{E}[A_1]
+ \mathbb{E}[B_1(2aA_2 + cB_2)] + \mathcal{O}(h^2) \nonumber\\
&= 2a\,\mathbb{E}[A_1]
+ 2a\,\mathbb{E}[B_1]\mathbb{E}[A_2]
+ c\,\mathbb{E}[B_1B_2] + \mathcal{O}(h^2) \nonumber\\
&= ah + ah + c\Bigl(1-\lambda(1-\cos(\phi))h+\mathcal{O}(h^2)\Bigr)
+ \mathcal{O}(h^2) = c + \big(2a-c\lambda(1-\cos(\phi))\big)h + \mathcal{O}(h^2).
\end{align}
Define
\begin{equation} 
r = \sqrt{\frac{L_1}{\kappa}} + \frac{1}{2} \min\left\{ \lambda\sin^2\phi - 2\sqrt{\frac{L_1}{\kappa}}, \; 2\lambda(1-\cos\phi)- \sqrt{\frac{L_1}{\kappa}} \right\}.
\end{equation}
Then $
\sqrt{\frac{L_1}{\kappa}} < r < \lambda\sin^2\phi-\sqrt{\frac{L_1}{\kappa}}$, and $  
r<2\lambda(1-\cos\phi)$. 
Let us choose $ b=1$, $c=r$, $a=\frac{r}{2}\lambda(1-\cos\phi)$. 
With this choice,
\begin{equation*}
\frac{L_1}{\Upsilon}b-\kappa c = \sqrt{\kappa L_1}-\kappa r = -\kappa \left(r-\sqrt{\frac{L_1}{\kappa}}\right)<0.
\end{equation*}
Hence $\alpha = a- \kappa \left(r-\sqrt{\frac{L_1}{\kappa}}\right)h+\mathcal{O}(h^2).
$
Therefore, there exists $\hat{\alpha}>0$ such that for sufficiently small $h$
\begin{equation}
\alpha\leq a-\hat{\alpha}h.
\end{equation}
Similarly, $b(\lambda\sin^2\phi-\Upsilon)-c = \lambda\sin^2\phi-\sqrt{\frac{L_1}{\kappa}}-r>0$. Thus there exists $\hat{\beta}>0$ 
 for sufficiently small $h$, the following holds:
\begin{equation}
\beta\leq 1-\hat{\beta}h.
\end{equation}
Finally, $2a-c\lambda(1-\cos\phi) = r\lambda(1-\cos\phi)-r\lambda(1-\cos\phi) = 0$. Therefore, $\delta=c+\mathcal{O}(h^2)$ implying $|\delta-c|\leq Ch^2 $ for sufficiently small \(h\).
Denoting
\begin{align}
z_0=(x_0,v_0)^\top,
\qquad
\hat{P} =
\begin{pmatrix}
a & c/2 \\
c/2 & b
\end{pmatrix},
\end{align}
we can write $\mathcal{V}(x_0,v_0)=z_0^\top \hat{P}z_0$. Since $ a=\frac{r}{2}\lambda(1-\cos\phi)$, $\frac{c^2}{4b}=\frac{r^2}{4}$,
and $r<2\lambda(1-\cos\phi)$, we obtain
\begin{align}
a =\frac{r}{2}\lambda(1-\cos\phi)>\frac{r^2}{4} = \frac{c^2}{4b}.
\end{align}
Therefore, $\hat{P}\succ0 $. Using the bounds for $\alpha,\beta,\delta$, we get
\begin{equation}
\alpha|x_0|^2+\beta|v_0|^2+\delta(x_0\cdot v_0)
\leq
z_0^\top \hat{P} z_0
-
h z_0^{\top} \hat{Q}_h z_0,\text{ where }
\hat{Q}_h=
\begin{pmatrix}
\hat{\alpha} & -\dfrac{\delta-c}{2h} \\
-\dfrac{\delta-c}{2h} & \hat{\beta}
\end{pmatrix}.
\end{equation}
Since $|\delta-c|\leq Ch^2$, and $\hat{\alpha}, \hat{\beta}>0$ are independent of $h$, we get
$\det( \hat{Q}_h)= \hat{\alpha}\hat{\beta} - \left(\frac{\delta-c}{2h}\right)^2 > \frac{1}{2}\hat{\alpha}\hat{\beta}$ for sufficiently small $h$, thus it follows that $\hat{Q}_h\succ0 $. Since $\hat{P}\succ0$ and $\hat{Q}_h\succ0$, there exists an $\eta>0$ such that
\begin{equation}
\hat{Q}_h\succeq \eta \hat{P}
\end{equation}
for sufficiently small $h$. Hence $\alpha|x_0|^2+\beta|v_0|^2+\delta(x_0\cdot v_0)
\leq (1-\eta h)\mathcal{V}(x_0,v_0)
$. Therefore,
\begin{equation}
\mathbb{E}[\mathcal{V}(x_1,v_1)\mid x_0,v_0]
\leq
(1-\eta h)\mathcal{V}(x_0,v_0)+C_0.
\end{equation}
Using $C_0\le \widetilde C_0h$, we get
\begin{equation}
\mathbb{E}[\mathcal{V}(x_1,v_1)\mid x_0,v_0]
\leq
(1-\eta h)\mathcal{V}(x_0,v_0)+\widetilde C_0 h.
\end{equation}

For the SEBES scheme, the proof is the same as for EBE after conditioning on the randomness generated by the first switching step. Indeed, the \(\mathrm{S}\) steps only update the discrete switching variable and do not directly change $(x,v)$; moreover, the final $\mathrm{S}$ step occurs after \((x_1,v_1)\) has already been computed and hence does not affect \(\mathcal V(x_1,v_1)\). Conditional on the first switching update, the continuous part of SEBES is exactly an EBE step with the force $\nabla U(x_{1/2}\; ;\s_1)$ in the \(\mathrm{B}\) step. Since Assumptions~\ref{assump_1_well_posed}-\ref{assump_2_ergo} hold uniformly over the switching state, the same bounds used in the EBE proof apply uniformly, namely
\begin{equation}
|\nabla U(x;\s_1)|^2\le L_1|x|^2+L_2,
\quad
-(x\cdot\nabla U(x;\s_1))\le -\kappa |x|^2+\ell.
\end{equation}
Hence the same choices of $a,b,c$ and the same estimates for
$\hat\alpha,\hat\beta$ apply uniformly in the switching state. The
positive-definiteness argument for $\hat{P}$ and $\hat{Q}_h$ therefore gives
$\hat{Q}_h\succeq \eta_0 \hat{P}$ for all sufficiently small $h$, with $\eta_0$
independent of the switching state. Consequently,
\begin{equation}
\mathbb E[\mathcal V(x_1,v_1)\mid x_0,v_0,\s_0]
\leq
(1-\eta_0h)\mathcal V(x_0,v_0)+\widetilde C_0h.  
\end{equation}

Thus, the SEBES scheme satisfies the same Lyapunov drift condition as EBE under the same assumptions.

\end{proof}

By applying Harris theorem \cite{hairer2011yet} on fixed time skeleton kernel $\mathcal{Q} = (\mathcal{P}_h)^{K_\varepsilon}$, we obtain the desired result. 

\subsection{Proof of Theorem~\ref{sebes_error_expansion_thrm}}

Consider a measurable function $\varphi$ defined on $\mathbb{R}^{2d} \times \mathcal{M}$ and bounded by $1 +\cV$. Integrating \eqref{dis_Poisson_equation_SEBES} with respect to $\rho$ yields
\begin{align}
  \sum_{i\in\mathcal M} \int_{\mathbb R^{2d}} \left(\frac{\mathcal{P}_h-I}{h}\right)\Psi(x, v,i) \rho(x,v,i) \der x \der v    = \mu(\varphi) - \mu_h(\varphi).
\end{align}
Because $\mathcal{P}_h$ is a positive operator and the Poisson solution satisfies
\begin{align}
|\Psi(x,v,i)|\leq |\Psi|_{L^\infty_\cV}(1+\cV(x,v)),
\end{align}
we may use the Lyapunov drift condition $\mathcal{P}_h(1+\cV)\le c_1 \cV+c_2 $
to bound the pairing against the invariant density on $\mathbb R^{2d}\times\mathcal M$:
\begin{align}
\sum_{i\in\mathcal M}\int_{\mathbb R^{2d}} |\mathcal{P}_h\Psi(x,v,i)|\,\rho(x,v,i)\der x\der v
\nonumber &\leq
|\Psi|_{L^\infty_\cV}
\sum_{i\in\mathcal M}\int_{\mathbb R^{2d}} \big(\mathcal{P}_h(1+\cV)\big)(x,v,i)\,\rho(x,v,i)\der x\der v \nonumber  \\ 
&\leq 
|\Psi|_{L^\infty_\cV}
\sum_{i\in\mathcal M}\int_{\mathbb R^{2d}} \big(c_1\cV(x,v)+c_2\big)\,\rho(x,v,i)\der x\der v
<\infty, \label{swi_eqn_new_5.55}
\end{align}
where the finiteness follows from the exponential decay of $\rho$. The adjoint $\mathcal{P}^*_{h}$ is defined via duality
\begin{align}
   \sum_{i\in\mathcal M}\int_{\mathbb R^{2d}}  f(x,v,i) \mathcal{P}_h^*\rho(x,v,i) \der x \der v = \sum_{i\in\mathcal M}\int_{\mathbb R^{2d}} \mathcal{P}_h f(x,v,i) \rho(x,v,i) \der x \der v, 
\end{align}
where $f $ belongs to a class of bounded and measurable functions. Since each substep of $\mathrm{SEBES}$ maps absolutely continuous measures to absolutely continuous measures, $\mathcal{P}^*\rho$ can be identified as density.   Because of \eqref{swi_eqn_new_5.55}, Fubini's theorem can be used to extend duality for $f = \Psi$, and therefore, we get $\mathcal P^{*}_h \rho \in L^{1, \cV}$ and
\begin{align}\label{swi_eq_adjoint}
  \sum_{i\in\mathcal M} \int_{\mathbb R^{2d}} \Psi(x,v,i) \left(\frac{\mathcal{P}_h^*-I}{h}\right)\rho(x,v,i) \der x \der v 
   =\mu(\varphi) - \mu_h(\varphi).
\end{align}
  For the $\rB$-step, Taylor's theorem along the explicit flow yields
\begin{equation}
    P^{\mathrm{B}}_sf =
f+s\rB f+\frac{s^2}{2}\rB^2f+\frac{s^3}{6}\rB^3f
+\int_0^s\frac{(s-r)^3}{3!}P_r^\rB \mathrm{B}^4f \der r.
\end{equation}
 The function $\rE^j f$, $0<j \leq 4$ is smooth and has at most polynomial growth. It is not difficult to show the moment bounds for $\mathrm{SEBES}$ under Assumptions~\ref{assump_1_well_posed}-\ref{assump_2_ergo} following the similar set of arguments as used in establishing Lyapunov drift condition and therefore, we have 
\begin{equation*}
\sup_{i\in \mathcal{M}}\mathbb{E}\int_0^h |\rE^{j+1}f(X_s,V_s, i)|\der s<\infty, \quad 0<j\leq 4,
\end{equation*}
therefore, we get
\begin{align}
P^{\mathrm{E}}_sf
=
f+ s \mathrm{E}f+\frac{s^2}{2}\mathrm{E}^2f+\frac{s^3}{6}\mathrm{E}^3f + \int_0^{s}\frac{(s-r)^3}{3!} P_r^\mathrm{E} \mathrm{E}^4 f\der r,\quad 0\leq s \leq h.
\end{align}
Denote
\begin{align}
\mathrm{R}_\mathrm{E} f
= \int_0^{h/2}\frac{(h/2-s)^3}{3!}\,P_s^\mathrm{E} \mathrm{E}^4 f\,\der s,
\qquad \mathrm{R}_\mathrm{B} f =\int_0^{h}\frac{(h-s)^3}{3!}\,P_s^\mathrm{B} \mathrm{B}^4 f\,\der s.
\end{align}
For a given $(x,v) \in \mathbb{R}^d$, we treat $ i \to f(x,v, i)$ as $\mathbb{R}^m$ valued function. We apply Taylor's theorem to $P_s^{\rS} := e^{s Q}$ to get
\begin{equation}
P_{s}^\rS f = f+ s\rS f+\frac{s^2}{2}\rS^2 f+\frac{s^3}{6}\rS^3 f+\mathrm{R}_\rS f,
\end{equation}
with $\rS := Q(x)$ and $ \mathrm{R}_\rS f
=
\int_0^{s}\frac{(s-r)^3}{3!}P_r^\rS \rS^4 f \der r$. Consider the symmetric splitting
\begin{equation}
\mathcal{P}_h :=
P_{h/2}^\rS P_{h/2}^\rE P_h^\rB P_{h/2}^\rE P_{h/2}^\rS. 
\end{equation}
Denoting $P_{s}^\mathrm{F}= \mathrm{T}_\mathrm{F}+ \mathrm{R}_\mathrm{F}$ where $\mathrm{T}_{\mathrm{F}} := \mathrm{I} + s\mathrm{F} + \frac{s^2}{2}\mathrm{F}^2 + \frac{s^3}{6}\mathrm{F}^3$ with $\mathrm{F} \in \{\rE, \rB, \rS\}$ and $s = h$ if $\mathrm{F} = \rB$ and $s = h/2$ otherwise. Then
\begin{align*}
\mathcal{P}_h = (\mathrm{T}_\rS+ \mathrm{R}_\rS)(\mathrm{T}_\rE+ \mathrm{R}_\rE)(\mathrm{T}_\rB+ \mathrm{R}_\rB)(\mathrm{T}_\rE+ \mathrm{R}_\rE)(\mathrm{T}_\rS+ \mathrm{R}_\rS).
\end{align*}
Expanding the product and collecting the terms with no remainder factor, we obtain
\begin{equation}
\rT_\rS \rT_\rE \rT_\rB \rT_\rE \rT_\rS 
=
I+h\mathcal{A} +\frac{h^2}{2}\mathcal A^2
+h^3\Bigl(\frac16\mathcal A^3+\mathcal S\Bigr)+ \mathrm{R}_1,
\end{equation}
where $ \mathcal A=\rE+\rB + \rS$ and $\mathrm{R}_1$ contains all terms of order at least $h^4$ arising from 
$\rT_\rS \rT_\rE \rT_\rB \rT_\rE \rT_\rS$. All the terms containing at least one of 
$\rR_\rS,\rR_\rE,\rR_\rB$ are grouped into $\rR_2$. Consequently, we have the following expansion for $f \in C_c^{\infty}(\mathbb{R}^{2d}\times \mathcal{M})$:
\begin{align}\label{eq:SEBES-expansion}
\mathcal{P}_h f = f + h\mathcal{A} f + \frac{h^2}{2}\mathcal{A}^2 f
+h^3\Big(\frac{1}{6}\mathcal{A}^{3}+ \mathcal{S}\Big)f
+\rR_1 f+  \rR_2 f.
\end{align}
where $\mathcal{S} = \frac{1}{2}[\rB,[\rB,\rE]]
-\frac1{24}[\rE,[\rE,\rB]] +\frac1{12}[\rB+\rE,[\rB+\rE,\rS]] -\frac{1}{24}[\rS,[\rS,\rB+\rE]]$. 
We now pair the above expansion against the Gibbs density $\rho$ in order to derive the corresponding adjoint expansion. Since $f\in C_c^\infty(\mathbb R^{2d}\times\mathcal M)$, all terms produced by the operators under consideration are smooth and have at most polynomial growth. Together with the fact that $\rho$ is smooth and its derivatives decay sufficiently fast, this implies that all duality pairings and remainder integrals below are absolutely integrable. Consequently, Fubini's theorem may be applied to interchange the order of integration, and integration by parts yields no boundary contributions (due to sufficiently fast decay of $\rho$ and its derivatives, see Assumption~\ref{swi_decay_assump}). The weighted spaces $L^{1,\cV}$ and $L_\cV^\infty$ are used only for the subsequent remainder estimates and final duality argument. With these preliminaries in
place, the adjoints can be obtained. For terms like $\rT_\rB$ it is straightforward to obtain its adjoint which we write below
\begin{align}
    \rT_\rB^* = I+h\rB^*+\frac{h^2}{2}(\rB^*)^2+\frac{h^3}{6}(\rB^*)^3.
\end{align}
Likewise, we can obtain adjoints for $\rT_\rE$, $\rT_\rS$ and $\rR_1$ via repeated application of integration by parts and Fubini's theorem. Likewise, for reminder terms, for example:
\begin{align}
\rR_\rS^*\rho
=
\int_0^{h/2}\frac{(h/2-s)^3}{3!}(\rS^*)^4 (P_s^\rS)^*\rho \der s.
\end{align}
It is more involved to get adjoint of one of the terms appearing in $\rR_2$, for example
\begin{align*}
(\rR_\rS \rR_\rE \rR_\rB \rT_\rE \rT_\rS)f = \int_0^{h/2}\int_0^{h/2}\int_0^h
 q(s,r,u)
P_s^\rS \rS^4 P_r^\rE \rE^4 P_u^\rB \rB^4 \rT_\rE \rT_\rS f
\der u \der r \der s,
\end{align*}
where $q(s , r, u) =(h/2-s)^3(h/2-r)^3(h-u)^3/(3!)^3$.  
Hence, by reversing the operator order, its adjoint is
\begin{align*}
(\rR_\rS \rR_\rE \rR_\rB \rT_\rE \rT_\rS)^*\rho
&=
\int_0^{h/2}\int_0^h\int_0^{h/2}
q(u,r,s)\rT_\rS^* \rT_\rE^* (\rB^*)^4 (P_r^\rB)^* (\rE^*)^4 (P_s^\rE)^* (\rS^*)^4 (P_u^\rS)^*\rho
\der s\der r\der u.
\end{align*}
All remaining adjoint terms are obtained in the same manner. Therefore, pairing \eqref{eq:SEBES-expansion} against $\rho$, we obtain
\begin{align}
\sum_{i\in\mathcal M}&\int_{\mathbb R^{2d}} (\mathcal P_h f)(x,v,i)\,\rho(x,v,i)\der x\der v
=
\sum_{i\in\mathcal M}\int_{\mathbb R^{2d}} f(x,v,i)\rho(x,v,i)\der x\der v \nonumber\\
&
+ h\sum_{i\in\mathcal M}\int_{\mathbb R^{2d}} f_i\mathcal{A}^*\rho_i\der x\der v
+\frac{h^2}{2}\sum_{i\in\mathcal M}\int_{\mathbb R^{2d}} f_i(\mathcal{A}^*)^2\rho_i\der x\der v 
+h^3\sum_{i\in\mathcal M}\int_{\mathbb R^{2d}} f_i\Big(\frac{1}{6}(\mathcal{A}^*)^3+\mathcal{S}^*\Big)\rho_i\der x\der v \nonumber\\
&
+\sum_{i\in\mathcal M}\int_{\mathbb R^{2d}} f_i\rR_1^*\rho_i\der x\der v
+\sum_{i\in\mathcal M}\int_{\mathbb R^{2d}} f_i\rR_2^*\rho_i\der x\der v,
\label{swi_eq_SEBES_adjoint_pairing}
\end{align}
where $\mathcal A^*=\rS^*+\rE^*+\rB^*$, $\mathcal{S}^* = \frac1{12}[\rB^*,[\rB^*,\rE^*]]
-\frac1{24}[\rE^*,[\rE^*,\rB^*]]
+\frac{1}{12}[\rB^*+\rE^*,[\rB^*+\rE^*,\rS^*]]
-\frac1{24}[\rS^*,[\rS^*,\rB^*+\rE^*]]$ and, for convenience, we have denoted $f_i := f(x,v, i)$ and $\rho_i := \rho(x,v,i)$. 
Thanks to the invariance of Gibbs
measure for exact dynamics $(\mathcal{A}^{*})^j \rho = 0$ under condition \eqref{swi_eqn_2.9}, \eqref{swi_eq_SEBES_adjoint_pairing} becomes
\begin{align}
\sum_{i\in\mathcal M}&\int_{\mathbb R^{2d}} (\mathcal{P}_h f)(x,v,i)\,\rho(x,v,i)\der x\der v
=
\sum_{i\in\mathcal M}\int_{\mathbb R^{2d}} f(x,v, i)\rho(x,v, i)\der x\der v \nonumber\\
&
+h^3\sum_{i\in\mathcal M}\int_{\mathbb R^{2d}} f_i\mathcal{S}^*\rho_i\der x\der v 
+\sum_{i\in\mathcal M}\int_{\mathbb R^{2d}} f_i\rR_1^*\rho_i\der x\der v
+\sum_{i\in\mathcal M}\int_{\mathbb R^{2d}} f_i\rR_2^*\rho_i\der x\der v. \label{swi_eqn_5.62}
\end{align}
Since the density $\rho$ has sufficiently fast decaying derivatives, applying any finite combination of the adjoint operators
$\rS^*,\rE^*,\rB^*$ produces functions which retain this decay. Since $\cV$ is quadratic, these functions belong to $L^{1,\cV}$. By passing the weighted norm inside the exact time-integral representations of the remainder terms $\rR_i^* \rho$, the bound below follows:
\begin{equation*}
|\rR_i^*\rho|_{L^{1,\cV}}\le Ch^4,
\qquad i=1,2.
\end{equation*}
And, similarly $\mathcal{S}^*\rho\in L^{1,\cV}$. To construct a pairing with the Poisson solution $\Psi$, we approximate $\Psi$ with compactly supported functions. To this end, let $\chi_n\in C_c^\infty(\mathbb R^{2d})$ satisfy
$0\le \chi_n\le 1$ and $\chi_n\to 1$ point-wise as $n \to \infty$, and define
\begin{equation}
g_n(x,v,i):=\chi_n(x,v) \Psi(x,v,i).
\end{equation}
As our Lyapunov function is locally integrable and since $|\Psi|\le C(1+\cV)$ for all $i\in \mathcal{M}$, we have $g_n\in L^1(\mathbb{R}^{2d}\times \mathcal{M})$. 
Let $\eta\in C_c^\infty(\mathbb R^{2d})$ be a standard mollifier with unit mass, and define $
\eta_\delta(z):=\delta^{-2d}\eta(z/\delta)$. For each $n$, choose $\delta_n>0$  such  that for every $i \in \mathcal{M}$ (see \cite[Theorem~4.22]{brezis2011functional}), we have 
\begin{equation}\label{swi_eqn_5.64}
|\eta_{\delta_n}*g_n-g_n|_{L^1(\mathbb R^{2d})}<2^{-n},
\end{equation}
Define $ f_n(x,v,i):=(\eta_{\delta_n}*g_n(\cdot, \cdot, i))(x,v)\in C_c^\infty(\mathbb R^{2d}\times \mathcal{M}).$
Then, due to Tonelli's theorem, \eqref{swi_eqn_5.64} and finiteness of $\mathcal{M}$, we get
\begin{align}
\int_{\mathbb{R}^{2d}}\sum_{n=1}^\infty \sum_{i \in \mathcal{M}} |f_n-g_n| \der y = \sum_{n=1}^{\infty}\sum_{i \in \mathcal{M}} |f_n - g_n|_{L^1(\mathbb R^{2d})} < \infty,
\end{align}
hence, $ f_n(x,v,i)-g_n(x,v,i)\to0$ for a.e.  $(x,v) \in \mathbb{R}^{2d}$ and for every $i \in \mathcal{M}$. Since $g_n=\chi_n\Psi\to\Psi$ pointwise, we have established
\begin{equation}\label{eq_SEBES_fn_to_Psi}
f_n(x,v,i)\to \Psi(x,v,i)
\qquad\text{for a.e. } (x,v) \in \mathbb R^{2d} \text{ and every } i\in \mathcal{M}.
\end{equation}
We next show that $(f_n)$ is uniformly bounded by $1+\cV$. 
Since $\Psi\in L^\infty_\cV$, $|\Psi(y,i)|\leq |\Psi|_{L^\infty_\cV}(1+\cV(y))$.  Hence, for $z=(x,v)$,
\begin{align*}
|f_n(z, i)| &= \left| \int_{\mathbb{R}^{2d}}\eta_{\delta_n}(z-y)\chi_n(y)\Psi(y, i)\der y
\right| \leq
|\Psi|_{L^\infty_\cV}
\int_{\mathbb R^{2d}}\eta_{\delta_n}(z-y)(1+\cV(y))\der y.
\end{align*}
Since $\eta_{\delta_n}$ is compactly supported and $\delta_n\leq 1$, whenever
$\eta_{\delta_n}(z-y)\neq 0$ one has $|z-y|\le1$. For quadratic $\cV$, it is not difficult to show that there exists $C>0$ such that $1+\cV(y)\le C(1+\cV(z))$ whenever $|z-y|\leq 1$.
Therefore, $ |f_n(z,i)|\le C |\Psi|_{L^\infty_\cV}(1+\cV(z))$. Thus for every $\varrho\in L^{1,\cV}$,
\begin{equation}
|f_n\varrho| \leq C |\Psi|_{L^\infty_\cV}(1+\cV)|\varrho|
\in L^1(\mathbb{R}^{2d}),
\end{equation}
Therefore, using \eqref{eq_SEBES_fn_to_Psi} and dominated convergence theorem, we obtain for every $\varrho \in L^{1, \cV}$
\begin{equation}\label{eq:SEBES-pairing-limit}
\sum_{i \in \mathcal{M}}\int_{\mathbb R^{2d}} f_n (x,v,i) \varrho(x,v,i)\der x \der v 
\longrightarrow
\sum_{i \in \mathcal{M}} \int_{\mathbb R^{2d}} \Psi(x,v,i) \varrho(x,v,i) \der x \der v,
\end{equation}
as $n \to \infty$. As we already discussed that functions $
\varrho=\frac{P_h^*-I}{h}\rho_i$, $\varrho=\mathcal S^*\rho_i$ or $\varrho=R_j^*\rho_i$, $j =1, 2$ belong to $L^{1,\cV}$, therefore  applying \eqref{eq:SEBES-pairing-limit}, we may pass to the limit in the expansion \eqref{swi_eqn_5.62} and conclude that
\begin{align}
\sum_{i \in \mathcal{M}}\int_{\mathbb R^{2d}} \Psi(x,v,i) \left(\frac{\mathcal{P}_h^*-I}{h}\right)\rho(x,v,i)\der x\der v &= h^2\sum_{i \in \mathcal{M}}\int_{\mathbb R^{2d}} \Psi(x,v,i)\mathcal{S}^*\rho(x,v,i)\der x\der v \nonumber  \\  &
+\frac{1}{h}\sum_{i \in \mathcal{M}} \sum_{j=1,2}\int_{\mathbb R^{2d}} \Psi(x,v,i) \rR_j^*\rho(x,v,i)\der x\der v.
\label{eq_SEBES_strong_pairing}
\end{align}
Finally, since $|R_j^*\rho|_{L^{1,\cV}}\le Ch^4$, $j=1,2$ and $\Psi\in L^\infty_\cV$, then using \eqref{swi_eq_adjoint}, we get
\begin{equation}
\mu(\varphi) - \mu_h(\varphi)
=
h^2\sum_{i\in \mathcal{M}}\int_{\mathbb{R}^{2d}} \Psi(x,v,i) \mathcal{S}^*\rho (x,v,i)\der x\der v + \mathcal{O}(h^3).
\end{equation}
This completes the proof. 

\section{Numerical illustration}
In this section, we present a numerical experiment validating the second order convergence of SEBES scheme. With $\hat{\varphi}_{N}$, we denote the Monte Carlo approximation of $\mathbb{E}\varphi(X_{N})$,  i.e. $ \hat{\varphi}_{N} = \frac{1}{M}\sum_{i=1}^{M} \varphi(X^{(i)}_{N})$ where $X^{(i)}_{N}$ represents
 independent samples of the numerical solution.  We borrow this experiment from \cite{tretyakov2025sampling} to compare with switching Langevin algorithm. We consider a mixture of two-dimensional Gaussian distributions:
\begin{align*}
    \rho(x)
    =     \frac{1}{\mathcal{Z}}
    \left[\alpha_1 \exp\left(
    -\frac{1}{2}(x-m_1)^{T}\Sigma_1^{-1}(x-m_1)
    \right) + \alpha_2 \exp\left(
    -\frac{1}{2}(x-m_2)^{T}\Sigma_2^{-1}(x-m_2) \right)
    \right],
    \quad x \in \mathbb{R}^2 ,
\end{align*}
where $\mathcal{Z}$ denotes the normalizing constant, $\alpha_1 = 0.7$, $\alpha_2 = 0.5$, $ m_1 =  (1, 1)$, $m_2 = (-2, -1)$, 
and
\begin{align}
    \Sigma_1  &=
    \begin{pmatrix}
    2.0 & 0.1 \\
    0.1 & 0.5
    \end{pmatrix}, &
    \Sigma_2 &=
    \begin{pmatrix}
    1.0 & -0.1 \\ -0.1 & 1.0
    \end{pmatrix}.
\end{align}
In this experiment, we take $T = 200$ for both SEBES and switching Langevin algorithm. For SEBES, we choose $\lambda = 1$ and $\phi = \pi/4$.

\begin{table}[H]
\centering
\caption{Comparison of errors for SEBES \eqref{swi_hmc_sebes} and switching Langevin algorithm \eqref{swi_lang_algo_1}-\eqref{swi_lang_algo_2} with varying step sizes $h$. Number of trajectories for Monte Carlo estimation are $10^8$ for SEBES and $10^6$ for switching Langevin algorithm.}
\label{swi_hmc_experiment_1}
\begin{tabular}{l c c c c}
\toprule
 & \multicolumn{2}{c}{SEBES} & \multicolumn{2}{c}{Switching Langevin algo.} \\
\cmidrule(lr){2-3} \cmidrule(lr){4-5}
$h$ & weak error & MC error ($\times 10^{-4}$) & weak error & MC error ($\times 10^{-3}$) \\
\midrule
0.90 & 0.23046 & 8.72 & 0.60131 & 10.5 \\
0.81 & 0.18721 & 8.81 & 0.51139 & 10.4 \\
0.72 & 0.14758 & 8.90 & 0.43984 & 10.2 \\
0.64 & 0.11721 & 8.97 & 0.37165 & 10.1 \\
0.60 & 0.10322 & 9.00 & 0.33447 & 10.0 \\
0.57 & 0.09280 & 9.02 & 0.32556 &  9.98 \\
\bottomrule
\end{tabular}
\end{table}

\begin{figure}[H]
    \centering
    \includegraphics[width=0.5\linewidth]{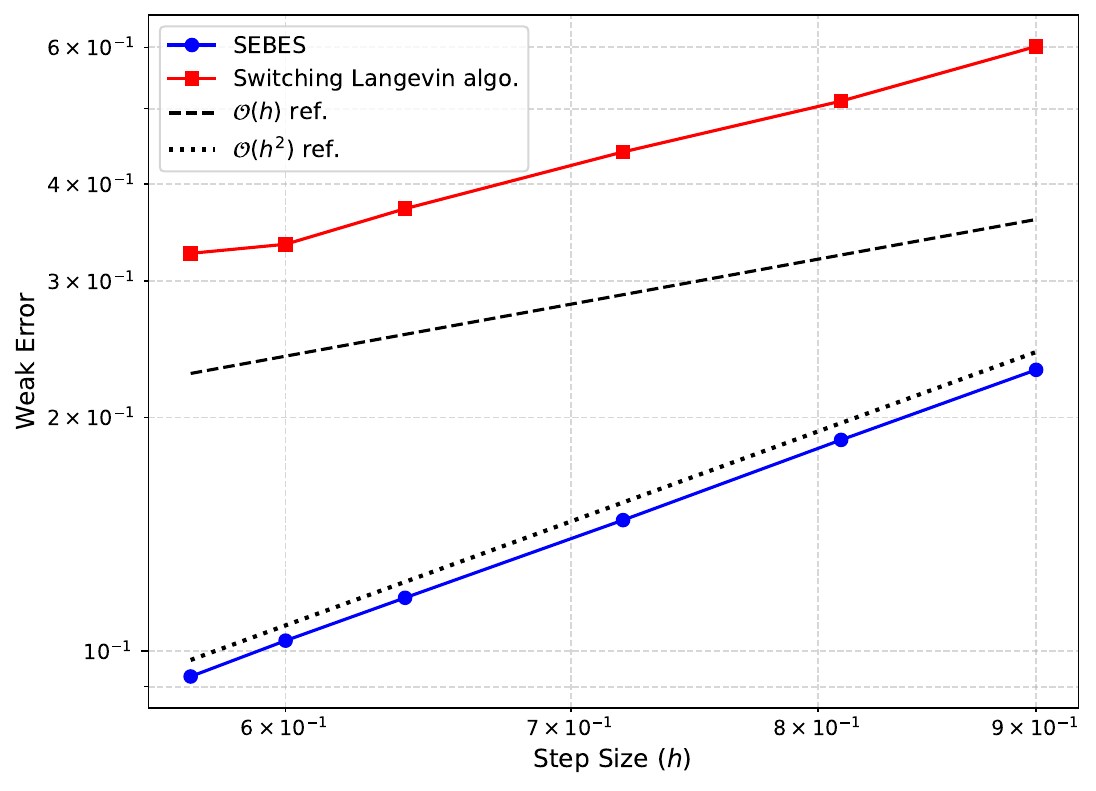}
    \caption{Convergence plot comparing SEBES and switching Langevin algorithm.}
    \label{2d_gaussian_comparison_exp}
\end{figure}

\section{Concluding remarks}
In this paper, we have considered splitting discretizations of continuous time Hamiltonian dynamics interlaced with Poisson refreshment and interacting with a regime-switching chain. We have proved second-order convergence in computing ergodic averages with respect to mixture probability distributions.  Since we do not have access to derivative bounds on the solution of Kolmogorov equations or Poisson equations which are first order partial integral differential equations in our case, we have developed a new approach based on discrete Poisson equation linked with numerical integrators. This methodology allows for larger class of test functions and therefore also providing convergence in total variation and $1$-Wasserstein distances. 

This work opens up several interesting avenues of future research: 
\begin{itemize}
    \item [(i)]  An important line of future work is to utilize the discrete Poisson equation based error analysis, presented in this paper, to obtain non-asymptotic bounds for numerical integrators for Langevin dynamics, randomized HMC and switching HMC under the assumption of strong convexity on potential function $U$ ($U_i$ in mixture distribution setting). 
    \item[(ii)] Another significant step for the switching HMC is to better understand the behavior of switching mechanism, perhaps via coupling techniques, so that an optimal choice of switching rates can be made. 
    \item[(iii)] There have been a lot of focus in developing and analyzing the numerical integrators for kinetic Langevin dynamics. Therefore, development of numerical integrators for randomized HMC is also an interesting topic of future work considering both are stochastic generalizations of deterministic Hamiltonian dynamics with significant differences in path-wise properties but serving the same purpose. 
\end{itemize}

\small
\bibliography{reference}
\bibliographystyle{alpha}
\end{document}